\documentclass[fleqn,10pt]{wlscirep}

\usepackage[sort&compress,square,comma,numbers]{natbib}

\usepackage{graphicx}
\usepackage{amsmath,amssymb}
\usepackage{amsfonts}
\usepackage{bm}
\usepackage{hyperref}
\usepackage{kbordermatrix}
\usepackage{amsthm}
\usepackage{subfigure}
\usepackage{sgame, tikz}

\title{Symmetric Decomposition of Asymmetric Games}

\author[1,2,*]{Karl Tuyls}
\author[1]{Julien Perolat}
\author[1]{Marc Lanctot}
\author[1]{Georg Ostrovski}
\author[2]{Rahul Savani}
\author[1]{Joel Leibo}
\author[3]{Toby Ord}
\author[1]{Thore Graepel}
\author[1]{Shane Legg}
\affil[1]{Google DeepMind, 6 Pancras Square, N1C 4AG London, UK}
\affil[2]{University of Liverpool, Dept. of Computer Science, Ashton Street, L69 3BX
Liverpool, UK}
\affil[3]{Oxford University, Faculty of Philosophy, Woodstock Road, OX2 6GG Oxford, UK}

\affil[*]{karltuyls@google.com}

\newtheorem{theorem}{Theorem}
\newtheorem{corollary}{Corollary}

\newtheorem*{definition}{Definition}
\newtheorem{property}{Property}

\newcommand{\R}{\mathbb{R}}

\keywords{Asymmetric Games, Replicator Dynamics, Counterpart Games}

\begin{abstract}
We introduce new theoretical insights into two-population asymmetric games allowing for an elegant symmetric decomposition into two single population symmetric games. Specifically, we show how an asymmetric bimatrix game $(A,B)$ can be decomposed into its symmetric counterparts by envisioning and investigating the payoff tables ($A$ and $B$) that constitute the asymmetric game,  as two independent, single population, symmetric games. We reveal several surprising formal relationships between an asymmetric two-population game and its symmetric single population counterparts, which facilitate a convenient analysis of the original asymmetric game due to the dimensionality reduction of the decomposition. The main finding reveals that if $(x,y)$ is a Nash equilibrium  of an asymmetric game $(A,B)$, this implies that $y$ is a Nash equilibrium of the symmetric counterpart game determined by payoff table $A$, and $x$ is a Nash equilibrium of the symmetric counterpart game determined by payoff table $B$. Also the reverse holds and combinations of Nash equilibria of the counterpart games form Nash equilibria of the asymmetric game. We illustrate how these formal relationships aid in identifying and analysing the Nash structure of asymmetric games, by examining the evolutionary dynamics of the simpler counterpart games in several canonical examples.

\end{abstract}
\begin{document}

\flushbottom
\maketitle
\thispagestyle{empty}

%\noindent Please note: Abbreviations should be introduced at the first mention in the main text – no abbreviations lists. Suggested structure of main text (not enforced) is provided below.

\section*{Introduction}

%The Introduction section, of referenced text\cite{Figueredo:2009dg} expands on the background of the work (some overlap with the Abstract is acceptable). The introduction should not include subheadings.

We are interested in analysing the Nash structure and evolutionary dynamics of strategic interactions in multi-agent systems. Traditionally, such interactions have been studied using single population replicator dynamics models, which are limited to symmetric situations, i.e., players have access to the same set of strategies and the payoff structure is symmetric as well~\cite{BloembergenTHK15}. For instance, Walsh et al. introduce an empirical game theory methodology (also referred to as heuristic payoff table method) that allows for analysing multiagent interactions in complex multiagent games\cite{Walsh02,Walsh03}. This method has been extended by others and been applied e.g. in continuous double auctions, variants of poker and multi-robot systems\cite{TuylsP07,PonsenTKR09,Wellman06,BloembergenTHK15,PhelpsCMNPS07,PhelpsPM04,Lanctot17}. Similar evolutionary methods have been applied to the modelling of human cooperation, language, and complex social dilemma's\cite{Perc,Moreira13,Santos16,PerolatLZBTG17,LazaridouPB16b,DeVylder,Cho}. Though these evolutionary methods have been very useful in providing insights into the type and form of interactions in such systems, the underlying Nash structure, and evolutionary dynamics, the analysis is limited to symmetric situations, i.e., players or agents can be interchanged and have access to the same strategy set, in other words there are no different roles for the various agents involved in the interactions (e.g. a seller vs a buyer in an auction). As such this method is not directly applicable to asymmetric situations in which the players can choose strategies from different sets of actions, with asymmetric payoff structures. Many interesting multiagent scenarios involve asymmetric interactions though, examples include simple games from game theory such as e.g. the Ultimatum Game or the Battle of the Sexes and more complex board games that can involve various roles such as Scotland Yard, but also trading on the internet for instance can be considered asymmetric.

There exist approaches that deal with asymmetry in multiagent interactions, but they usually propose to transform the asymmetric game into a symmetric game, with new strategy sets and payoff structure, which then can be analysed again in the context of symmetric games. This is indeed a feasible approach, but not easily scalable to the complex interactions mentioned before, nor is it practical or intuitive to construct a new symmetric game before the asymmetric one can be analysed in full. The approach we take in this paper does not require constructing a new game and is theoretically underpinned, revealing some new interesting insights in the relation between the Nash structure of symmetric and asymmetric games.

Analysing multiagent interactions using evolutionary dynamics, or replicator dynamics, provides not only valuable insights into the (Nash) equilibria and their stability properties, but also sheds light on the behaviour trajectories of the involved agents and the basins of attraction of the equilibrium landscape\cite{Nowak06,Tuyls03,TuylsP07,DeVylder,BloembergenTHK15}. As such it can be a very useful tool to analyse the Nash structure and dynamics of several interacting agents in a multiagent system. However, when dealing with asymmetric games the analysis quickly becomes tedious, as in this case we have a coupled system of replicator equations, and changes in the behaviour of one agent immediately change the dynamics in the linked replicator equation describing the behaviour of the other agent, and vice versa. 
This paper sheds new light on asymmetric games, and reveals a number of theorems, previously unknown, that allow for a more elegant analysis of asymmetric multiagent games. The major innovation is that we decouple asymmetric games in their \textit{symmetric counterparts}, which can be studied in a symmetric fashion using symmetric replicator dynamics. The Nash equilibria of these symmetric counterparts are formally related to the Nash equilibria of the original asymmetric game, and as such provide us with a means to analyse the asymmetric game using its symmetric counterparts. Note that we do not consider asymmetric replicator dynamics in which both intra-species (within a population) and inter-species interactions (between different populations) take place~\cite{Cressman14}, but we only consider inter-species interactions in which two different roles interact, i.e., truly asymmetric games~\cite{Selten80}. 

One of our main findings is that the \textit{$x$ strategies} (player 1) and the \textit{$y$ strategies} (player 2) of a mixed Nash equilibrium of full support in the original asymmetric game, also constitute Nash equilibria in the symmetric counterpart games. The symmetric counterpart of player 1 ($x$) is defined on the payoff of player 2 and vice versa. We prove that for full support strategies, Nash equilibria of the asymmetric game are pairwise combinations of Nash equilibria of the two symmetric counterparts. Then, we show that this property stands without the assumption of full support as well. Though this analysis does not allow us to visualise the evolutionary dynamics of the asymmetric game itself, it does allow us to identify its Nash equilibria by investigating the evolutionary dynamics of the counterparts. As such we can easily distinguish Nash equilibria from other restpoints in the asymmetric game and get an understanding of its underlying Nash structure.

The paper is structured as follows: we first describe related work, then we continue with introducing essential game theoretic concepts. Subsequently, we present the main contributions and we illustrate the strengths of the theory by carrying out an evolutionary analysis on four canonical examples. Finally, we discuss the implications and provide a deeper understanding of the theoretical results.

\section*{Related Work}

The most straightforward and classical approach to asymmetric games is to treat agents as evolving separately: one population per player, where each agent in a population interacts by playing against agent(s) from the other population(s), i.e. co-evolution~\cite{Taylor79}. This assumes that players of these games are always fundamentally attached to one role and never need to know/understand how to play as the other player. In many cases, though, a player may want to know how to play as either player. For example, a good chess player should know how to play as white or black. This reasoning inspired the role-based symmetrization of asymmetric games~\cite{Guanersdorfer91}.

The role-based symmetrization of an arbitrary bimatrix game defines a new (extensive-form) game where before choosing actions the role of the two players are decided by uniform random chance. If two roles are available, an agent is assigned one specific role with probability $\frac{1}{2}$. Then, the agent plays the game under that role and collects the role-specific payoff appropriately. A new strategy space is defined, which is the product of both players' strategy spaces, and a new payoff matrix computing (expected) payoffs for each combination of pure strategies that could arise under the different roles. There are relationships between the sets of evolutionarily stable strategies and rest points of the replicator dynamics between the original and symmetrized game\cite{Cressman03,Cressman14}.

This single-population model forces the players to be general: able to devise a strategy for each role, which may unnecessarily complicate algorithms that compute strategies for such players.
In general, the payoff matrix in the resulting role-based symmetrization is $n!$ ($n$ being the number of agents) times larger due to the number of permutations of player role assignments.
There are two-population variants that formulate the problem slightly differently: a new matrix that encapsulates both players' utilities assigns 0 utility to combinations of roles that are not in one-to-one correspondence with players\cite{Accinelli2011}. This too, however, results in an unnecessarily larger (albeit sparse) matrix.

Lastly, there are approaches that have structured asymmetry, that arises due to ecological constraints such as locality in a network and genotype/genetic relationships between population members~\cite{Avoy15}. Similarly here, replicator dynamics and their properties are derived by transforming the payoff matrix into a larger symmetric matrix.

Our primary motivation is to enable analysis techniques for asymmetric games. However, we do this by introducing new {\it symmetric counterpart dynamics} rather than using standard dynamics on a symmetrised game. Therefore, the traditional role interpretation as well as any method that enlarges the game for the purpose of obtaining symmetry is unnecessarily complex for our purposes. Consequently, we consider the original co-evolutionary interpretation, and derive new (lower-dimensional) strategy space mappings.

\section*{Preliminaries and Methods}\label{sec:prelim}

In this section we concisely outline (evolutionary) game theoretic concepts necessary to understand the remainder of the paper~\cite{Weibull97, Hofbauer98, Cressman03}. We briefly specify definitions of Normal Form Games and solution concepts such as Nash Equilibrium in a single population game and in a two-population game. Furthermore, we introduce the Replicator Dynamics (RD) equations for single and two population games and briefly discuss the concept of Evolutionary Stable Strategies (ESS) introduced by Smith and Price in 1973~\cite{Smith73,Zeeman80,Zeeman81}.
%ESS and RD are two key concepts from Evolutionary Game Theory (EGT), and the idea of an evolutionary stable strategy was introduced by John Maynard Smith and Price in 1973~\cite{Smith73,Zeeman80,Zeeman81}. We also introduce the new concept of counterpart RD.

\subsection*{Normal Form Games and Nash Equilibrium}

\begin{definition}
A two-player Normal Form Game (NFG) $G$ is a 4-tuple $G=(S_1,S_2,A,B)$, with pure strategy sets $S_1$ and $S_2$ for player 1, respectively player 2, and corresponding payoff tables $A$ and $B$. Both players choose their pure strategies (also called actions) simultaneously.
\end{definition}

The payoffs for both players are represented by a bimatrix $(A, B)$, which gives the payoff for the row player in $A$, and the column player in $B$ (see
Table \ref{fig:NFG} for a two strategy example). Specifically, the row player chooses one of the two rows, the column player chooses one of the columns, and the outcome of their joint strategy determines the payoff to both.

\begin{table}[tb] 
\centering
	\caption{General payoff bimatrix (A, B) for a two-player two-action normal form game, where player 1 can choose between actions $A_1$ and $A_2$, and player 2 can choose between actions $B_1$ and $B_2$.}\label{fig:NFG}  
\begin{game}{2}{2}[Player 1][Player 2]
	    &  $B_1$      &  $B_2$     \\
	 $A_1$  &  $a_{11}, b_{11}$ & $a_{12},b_{12}$  \\
	 $A_2$  &  $a_{21}, b_{21}$ & $a_{22}, b_{22}$\\
\end{game}
\end{table}

In case $S_1$=$S_2$ and $A=B^T$ the players are interchangeable and we call the game symmetric. In case at least one of these conditions is not met we have an asymmetric game.
In classical game theory the players are considered to be individually rational, in the sense that each player is perfectly logical trying to maximise their own payoff, assuming the others are doing likewise. Under this assumption, the Nash equilibrium (NE) solution concept can be used to study what players will reasonably choose to do.

\noindent We denote a strategy profile of the two players by the tuple $(x,y) \in \Delta S_1\times \Delta S_2$, where $\Delta S_1$, $\Delta S_2$ are the sets of mixed strategies, that is, distributions over the pure strategy sets or action sets. The strategy $x$ (respectively $y$) is represented as a vector in $\mathbb{R}^{|S_1|}$ (respectively $\mathbb{R}^{|S_2|}$) where each entry is the probability of playing the corresponding action.
%The associated payoff $f(\mathbf{s})=f_1(x)\times f_2(y)$, 
The payoff associated with player 1 is $x^TA y$ and $x^TB y$ is the payoff associated with player 2.
%and $f_2(y)$ for player 1 and player, respectively, i.e. $f_1(\mathbf{s}) = x^TA y$ ($f_2(\mathbf{s})=x^TB y$ for player 2). The payoff of the game for a given strategy profile $\mathbf{s}=(x,y)$ is the tuple $f(\mathbf{s})=(f_1(\mathbf{s}),f_2(\mathbf{s}))$.
A strategy profile $(x,y)$ now forms a NE if no single player can do better by unilaterally switching to a different strategy. In other words, each strategy in a NE is a best response against all other strategies in that equilibrium. Formally we have,

\begin{definition}
A strategy profile $(x,y)$ is a Nash equilibrium, iff the following holds: \\

$\forall x' \in \Delta S_1,\; x^T A y \geq x'^T A y$ and $\forall y' \in \Delta S_2,\;  x^T B y \geq x^T B y'$

\end{definition}

% If the inequality holds strictly for all players and alternative strategies then we call $s$ a strict Nash equilibrium. {\color{blue} Check this definition. By the way I think we should stick to the notations with $x,y,A,B$ and get ride of those with $s^i,S^i$}. 
In the following, we will write $NE(A,B)$ for the set of Nash equilibria of the game $G=(S_1, S_2, A, B)$. Furthermore, a Nash equilibrium is said to be pure if only one strategy of the strategy set is played and we will say that it is completely mixed if all pure strategies are played with a non-zero probability.

In evolutionary game theory, games are often considered with a single population. In other words, a player is playing against itself and only a single payoff table $A$ is necessary to define the game (note that this definition only makes sense when $|S_1|=|S_2|=n$). In this case, the payoff received by the player is $x^T A x$ and the following definition describes the Nash equilibrium:

\begin{definition}
In a single population game, a strategy $x$ is a Nash equilibrium, iff the following holds: \\

$\forall x', x^T A x \geq x'^T A x$
\end{definition}
In this single population case, we will write that $x \in NE(A)$. %Note that those two notions are related: $(x,x) \in NE(A,A^T)$ implies that $x \in NE(A)$, and the reverse is true in the case $A=A^T$, but not necessarily otherwise. It's important to be aware of the difference between both sets as the banach solver, which we use to compute and verify the Nash equilibria from the replicator dynamics, computes $NE(A,A^T)$ and not $NE(A)$ (see \url{http://banach.lse.ac.uk/}\cite{Avis10}). This means that our visualization might contain Nash equilibria that were not computed by the banach solver. We will see an example of this in the experimental illustration.

\subsection*{Replicator Dynamics}

Replicator Dynamics in essence are a system of differential equations that describe how a population of pure strategies, or replicators, evolve through time~\cite{Weibull97,Gintis09}. In their most basic form they correspond to the biological \emph{selection} principle, i.e. survival of the fittest. More specifically the \emph{selection} replicator dynamic mechanism is expressed as follows:

\begin{equation}\label{eq:singlerd}
 \frac{dx_i}{dt}=x_i[(A x)_i-x^TA x] 
\end{equation}

Each replicator represents one (pure) strategy $i$. This strategy is inherited by all the offspring of the replicator. $x_i$ represents the density of strategy $i$ in the population, $A$ is the payoff matrix which describes the different payoff
values each individual replicator receives when interacting with other replicators in the population. The state of the population $x$ can be described as a probability vector $x$ $= (x_1, x_2, ..., x_n)$ which expresses the different densities of all the different types of replicators in the population. Hence $(Ax)_i$ is the payoff which replicator $i$ receives in a population with state $x$ and $x^T Ax$ describes the average payoff in the population. The support $I_x$ of a strategy is the set of actions (or pure strategies) that are played with a non-zero probability $I_x = \{i \; | x_i > 0\}$.

In essence this equation compares the payoff a strategy receives with the average payoff of the entire population. If the strategy scores better than average it will be able to replicate \emph{offspring}, if it scores lower than average its presence in the population will diminish and potentially approach extinction. The population remains in the simplex ($\sum_i x_i = 1$) since $\sum_i \frac{dx_i}{dt} = 0$.

\subsection*{Evolutionary Stable Strategies}
Originally, an Evolutionary Stable Strategy was introduced in the context of a symmetric single population game~\cite{Gintis09,Smith73} (as introduced in the previous section), though this can be extended to multi-population games as well as defined in the next section\cite{Sandholm10,Cressman03}. Imagine a population of simple agents playing the same strategy. Assume that this population is invaded by a different strategy, which is initially played by a small proportion of the total population. If the reproductive success of the new strategy is smaller than the original one, it will not overrule the original strategy and will eventually disappear. In this case we say that the strategy is \emph{evolutionary stable} (ESS) against this newly appearing strategy. In general, we say a strategy is ESS if it is robust against evolutionary pressure from any appearing mutant replicator not yet present in the population (or only with a very small fraction).

\subsection*{Asymmetric Replicator Dynamics}

We have assumed replicators come from a single population, which makes the model only applicable to symmetric games. One can now wonder how the previous introduced equations extend to asymmetric games. Symmetry assumes that strategy sets and corresponding payoffs are the same for all players in the interaction. An example of an asymmetric game is the famous Battle of the Sexes (BoS) game illustrated in Table \ref{fig:BoS}. In this game both players do have the same strategy set, i.e., go to the opera or go to the movies, however, the corresponding payoffs for each are different, expressing the difference in preferences that both players have in their respective roles.

\begin{table}[ht] 
\centering
\begin{game}{2}{2}[][]
	    &  $O$      &  $M$     \\
	 $O$  &  $3,2$ & $0,0$  \\
	 $M$  &  $0,0$ & $2,3$\\
\end{game}
\caption{Payoff bimatrix for the Battle of the Sexes game. Strategies $O$ and $M$ correspond to going to the Opera and going to the Movies respectively.}\label{fig:BoS}
\end{table}

If we would like to carry out a similar evolutionary analysis as before we will now need two populations, one for each player over its respective strategy set, and we need to use the asymmetric or coupled version of the replicator dynamics, i.e.,

\begin{definition}
\begin{equation}\label{eq:P1rd}
 \frac{dx_i}{dt}=x_i[(A y)_i - x^TA y] \qquad \text{and} \qquad \frac{dy_i}{dt}=y_i[(x^T B)_i- x^TB y] 
\end{equation}
\end{definition}

\noindent with payoff tables $A$ and $B$, respectively for player 1 and 2. In case $A = B^T$ the equations reduce to the single population model.

\subsection*{Symmetric Counterpart Replicator Dynamics}

We now introduce a new concept, the \emph{symmetric counterpart} replicator dynamics (SCRD) of asymmetric replicator equations. We consider the two payoff tables $A$ and $B$ as two independent games that are no longer coupled, and in which both players participate. In the first counterpart game all players choose their strategy according to distribution $y$, the original strategy or replicator distribution for the 2nd population, or player 2, and in the second counterpart game all players choose their strategy according to distribution $x$, the original strategy or replicator distribution for the 1st population, or player 1. This gives us the following two sets of replicator equations:

%\begin{definition}
%\begin{equation}
%  \qquad \text{and} \qquad \frac{dx_i}{dt}=x_i[(x^TB)_i-x^TBx]
%\end{equation}
%\end{definition}

\begin{minipage}[b]{.3\textwidth}
\vspace{-\baselineskip}
\begin{equation}
\frac{dy_i}{dt}=y_i[(Ay)_i - y^TAy] \label{ESSDecoupledA} 
\end{equation}
\end{minipage}%
\qquad and \qquad
\begin{minipage}[b]{.3\textwidth}
\vspace{-\baselineskip}
\begin{equation}
\frac{dx_i}{dt}=x_i[(x^TB)_i-x^TBx] \label{ESSDecoupledB}
\end{equation}
\end{minipage} 

In the results Section we will introduce some remarkable relationships between the equilibria of asymmetric replicator equations and the equilibria of their symmetric counterpart equations, which facilitates, and substantially simplifies, the analysis of the Nash structure of asymmetric games.

\subsection*{Visualising evolutionary dynamics}
One can visualise the replicator dynamics in a directional field and trajectory plot, which provides useful information about the equilibria, flow of dynamics and basins of attraction. As long as we stay in the realm of 2-player 2-action games this can be achieved relatively easily by plotting the probability with which player 1 plays its first action on the x-axis, and the probability with which player 2 plays its first action on the y-axis. Since there are only 2 actions for each player, this immediately gives a complete image of the dynamics over all strategies, since the probability for the second action $a_2$ to be chosen is one minus the first. By means of example we show a directional field plot here for the famous Prisoner's dilemma game (game illustrated in Table \ref{fig:PD}).

\begin{table}[ht] 
\centering
	  %The use of the star * after caption is to remove the text "Table" from the title
\begin{game}{2}{2}[][]
	    &  $C$      &  $D$     \\
	 $C$  &  $3,3$ & $0,5$  \\
	 $D$  &  $5,0$ & $1,1$\\
\end{game}
\caption{Payoff matrix for the Prisoner's Dilemma game. Strategies $D$ and $C$ correspond to the actions \textit{Defect} and \textit{Cooperate}.}\label{fig:PD}
\end{table}

The directional field plot, and corresponding trajectories, are shown in Figure \ref{fig:dfieldPD}. For both players the axis represents the probability with which they play \textit{Defect} (D). As can be observed all dynamics are absorbed by the pure Nash equilibrium $(D,D)$ in which both players defect.

\begin{figure}[h!]
    \centering
    \includegraphics[width=6cm]{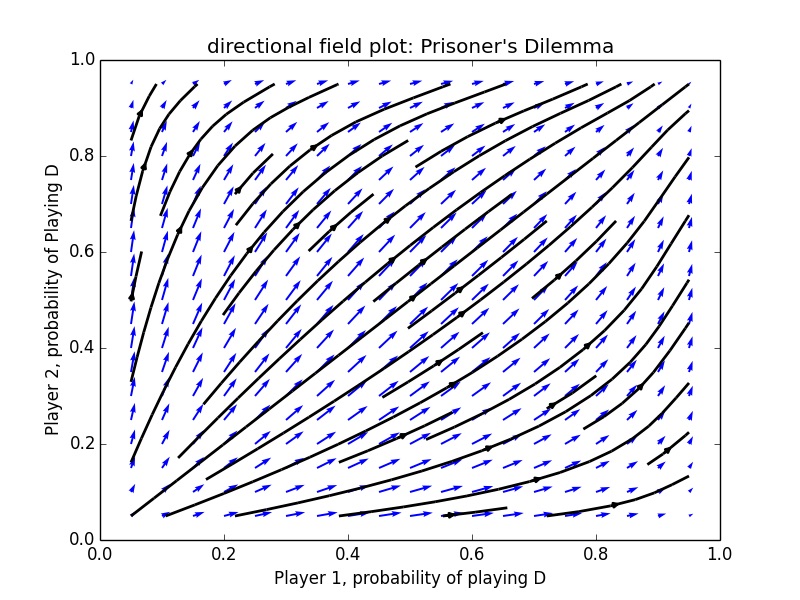}
    \caption{Legend. Directional field plot of the Prisoner's Dilemma game.}
    \label{fig:dfieldPD}
\end{figure}

Unfortunately, we cannot use the same type of plot illustrating the dynamics when we consider more than two strategies. However, if we move to single population games we can easily rely on a simplex plot. In the case of a two population game the situation become tedious as we will discuss later. Specifically, the set of probability distributions over $n$ elements can be represented by the set of vectors $(x_1,...,x_n)$ $\in \R^n$, satisfying $x_1,..., x_n \geq 0$ and $\sum_i x_i = 1$. This can be seen to correspond to an $n-1$-dimensional structure called a simplex $\Sigma_n$ (or simply $\Sigma$,
when $n$ is clear from the context).
In many of the figures throughout the paper we
use $\Sigma_3$, projected as an equilateral triangle. For example, consider the single population \textit{Rock-Paper-Scissors} game, described by the payoff matrix shown in Figure \ref{fig:RSP}a.

\begin{figure}[h!]
    \centering
    \includegraphics[width=12cm]{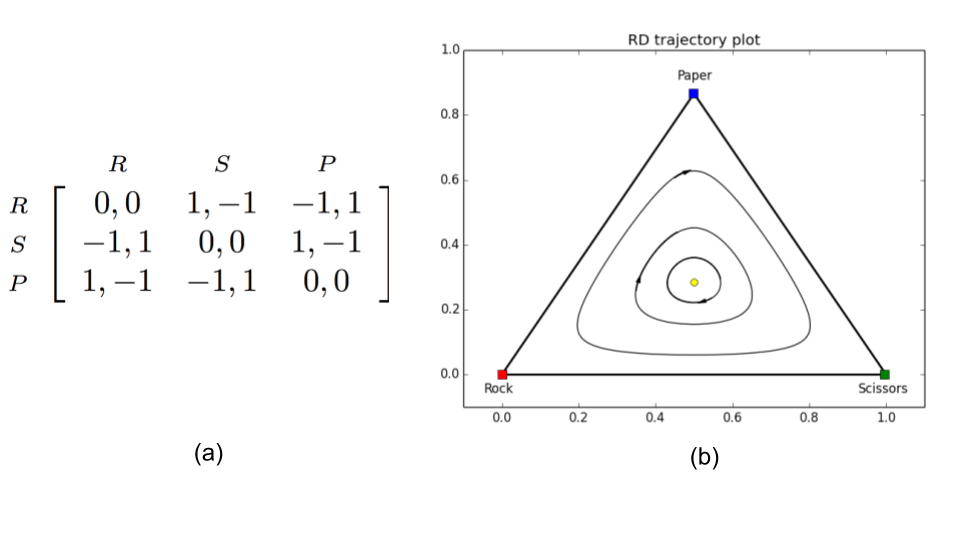}
    \caption{Legend. (a) Payoff matrix for the Rock-Paper-Scissors game. Strategies $R$, $S$ and $P$ correspond to playing respectively \textit{R}ock, \textit{S}cissors, \textit{P}aper. (b) $\Sigma_3$ Trajectory plot of the Rock-Paper-Scissors game. The Nash equilibrium is marked with a full yellow dot.}
    \label{fig:RSP}
\end{figure}

The game has one completely mixed Nash equilibrium, being $(\frac{1}{3},\frac{1}{3},\frac{1}{3})$. In Figure \ref{fig:RSP}b we have plotted the replicator equations $\Sigma_3$ trajectory plot for this game. Each of the corners of the simplex corresponds to one of the pure strategies, i.e., $\{Rock$, $Paper$, $Scissors\}$. For three strategies in the strategy simplex we then plot a trajectory illustrating the flow of the replicator dynamics. As can be observed from the plot, trajectories of the dynamics cycle around the mixed Nash equilibrium, which is not ESS and not asymptotically stable.

In fact, three categories of rest points can be discerned in single population replicator dynamics (see Figures~\ref{fig:ess}, \ref{fig:nash_not_ess}, \ref{fig:rest_point}). Figure~\ref{fig:ess} displays a stable Nash equilibrium called an Evolutionary Stable Strategy (ESS). An ESS is an attractor of the RD dynamical system defined in the previous section and has been one of the main foci of evolutionary game theory. The second type of rest points are the ones that are Nash but not ESS (Figure~\ref{fig:nash_not_ess}). These rest points are not an attractor of the RD but they have a specific form. Specifically, if a strategy is a Nash equilibrium, all the actions that are not part of the support are dominated, i.e., the support is invariant under the RD, which means that the fraction of a strategy cannot become non-zero if it is zero at some point.% Then, the RD cannot flow out of the support of the Nash equilibrium towards an action that is not part of the support. 
% , or .

%(note that there is no constraint on what appends on the support of the strategy and thus the equilibrium can be unstable as in Figure~\ref{fig:nash_not_ess}). 
%We call these a \emph{saddle point}. 
The third category that can occur is illustrated in Figure~\ref{fig:rest_point}. Those rest points are not Nash and thus there is an action outside of the support that is dominant. Thus, the flow will leave from points in the close vicinity of the rest point, which is called a \emph{source}.

%%%%%%%%%%%%%% Building that section

\begin{figure}[!tbp]
  \centering
  \begin{minipage}[c]{0.32\textwidth}
\includegraphics[width=6cm]{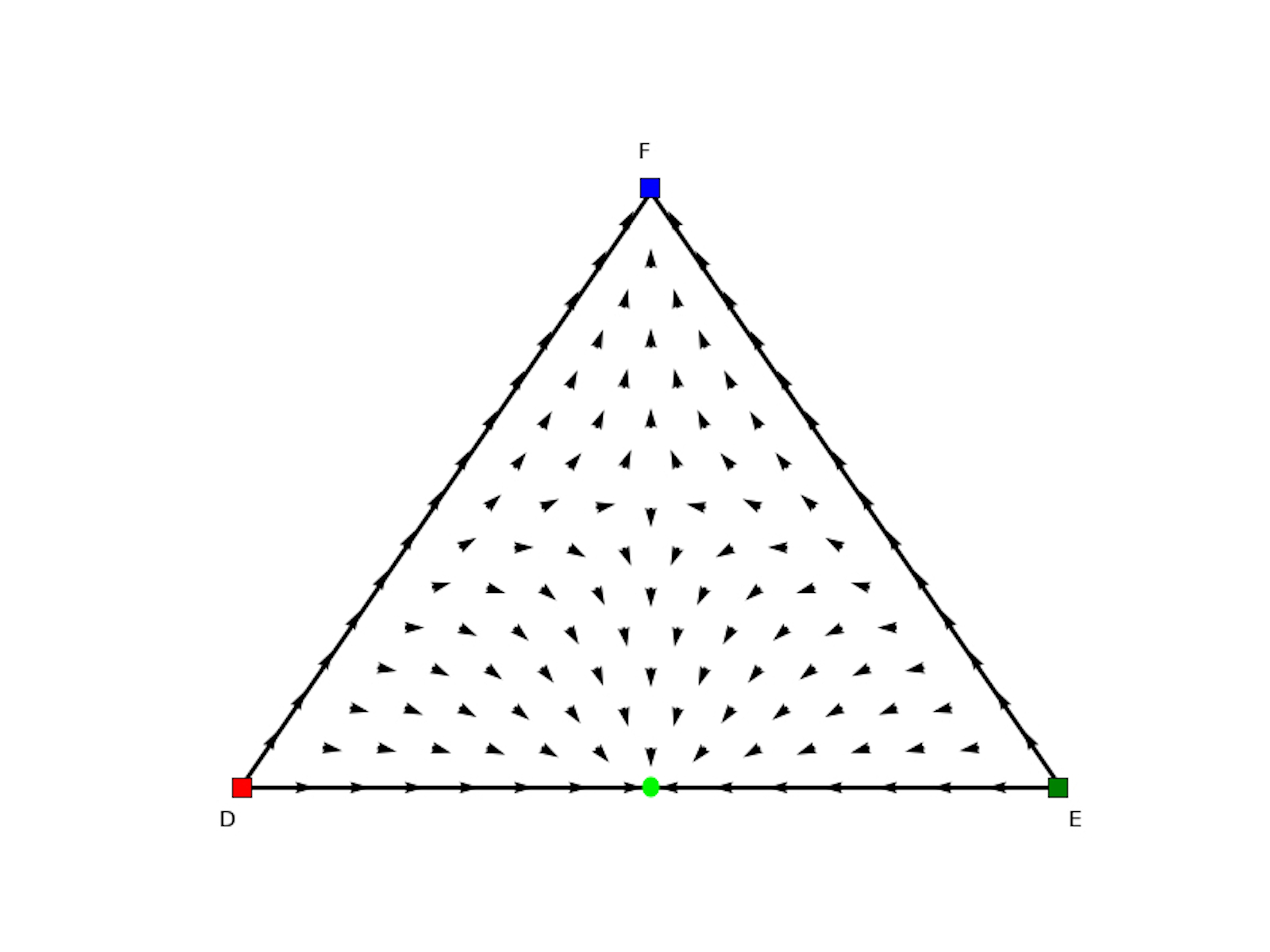}
   \caption{ESS}
    \label{fig:ess}
  \end{minipage}
 % \hfill
  \begin{minipage}[c]{0.32\textwidth}
\includegraphics[width=6cm]{{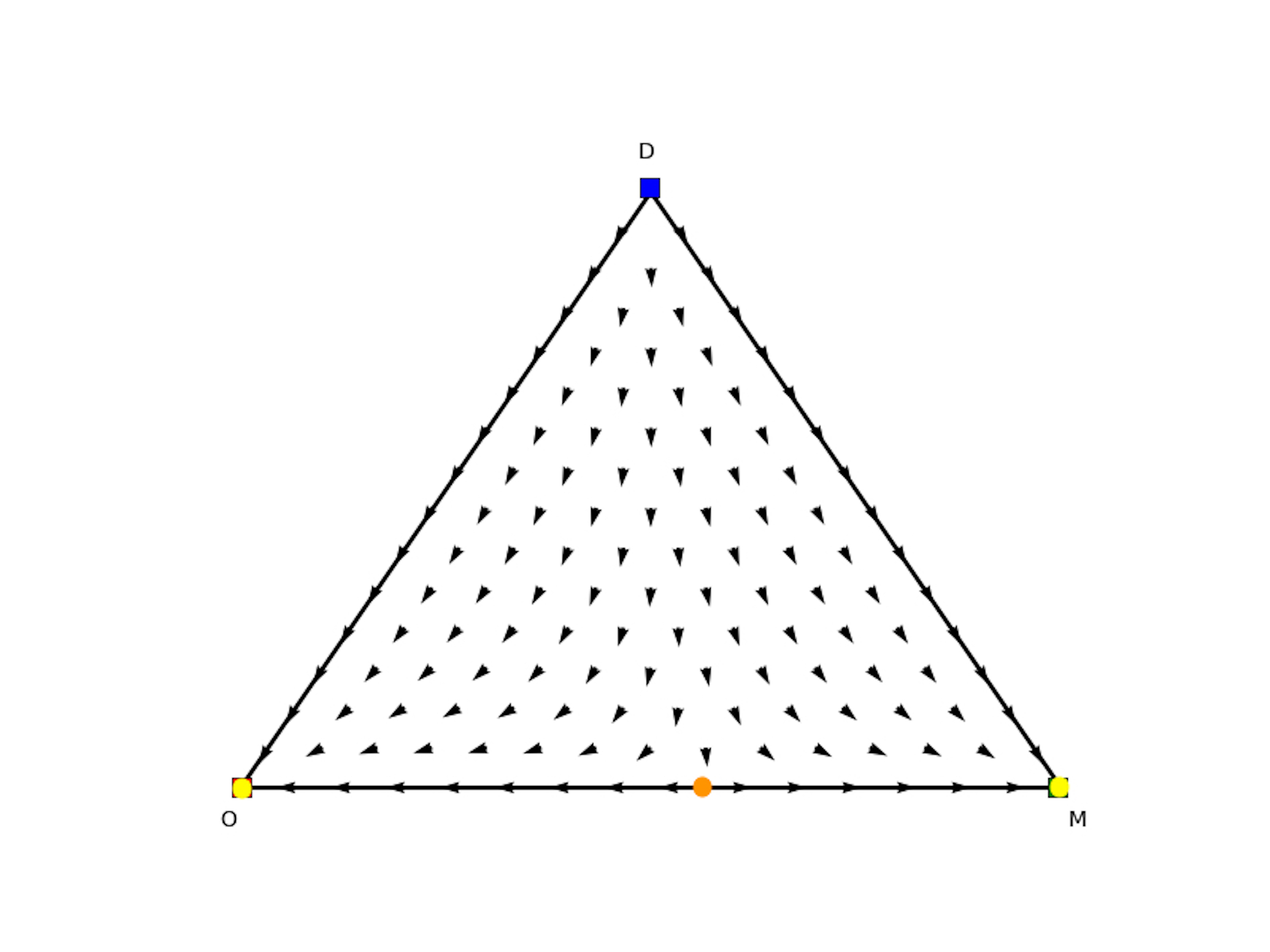}}
     \caption{NE but not ESS}
    \label{fig:nash_not_ess}
  \end{minipage}
 % \hfill
  \begin{minipage}[c]{0.32\textwidth}
\includegraphics[width=6cm]{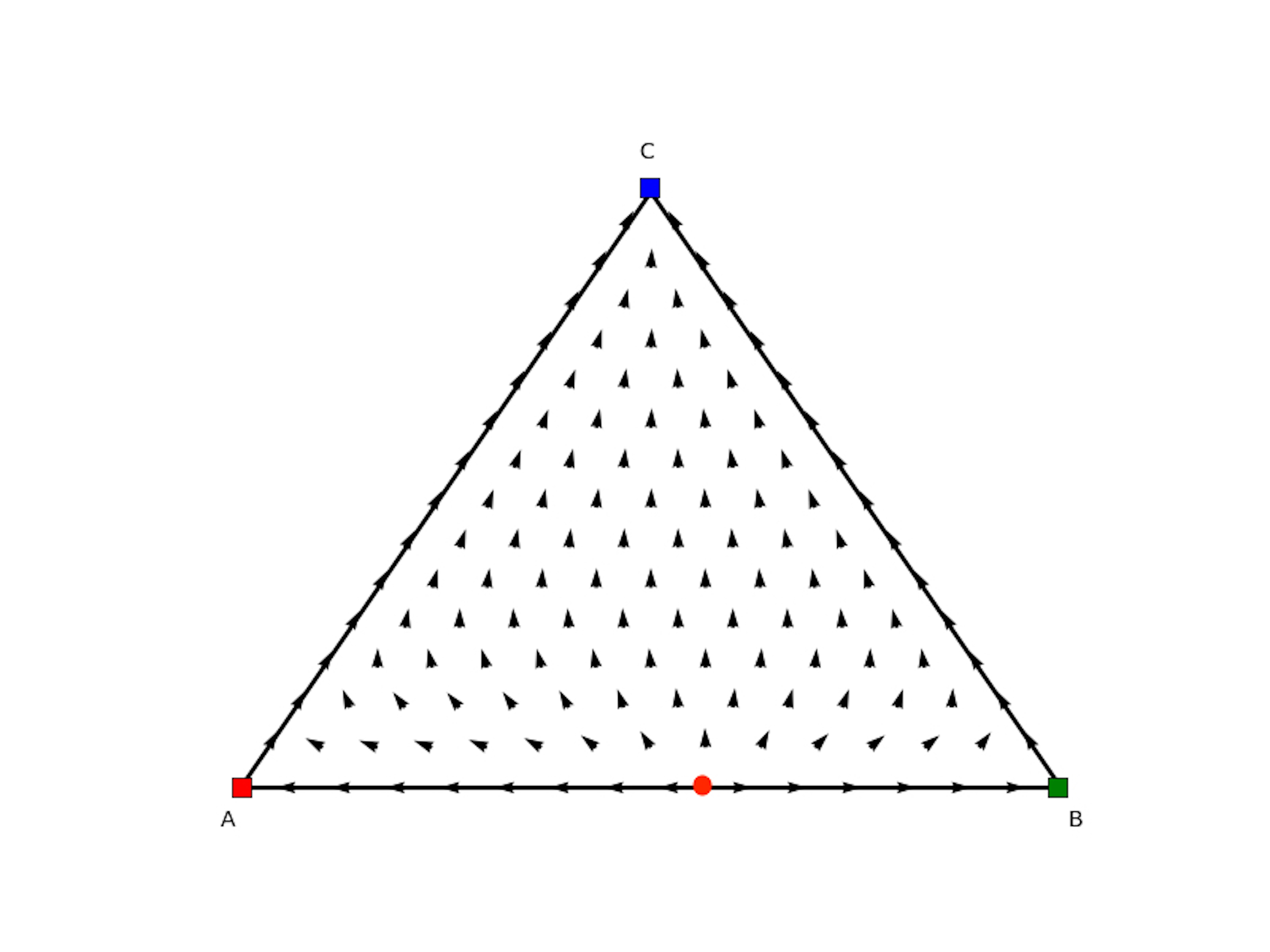}
   \caption{Rest point but not NE}
    \label{fig:rest_point}
  \end{minipage}
\end{figure}

\section*{Results}\label{sec:results}

In the following, we first present our main findings, formally relating Nash equilibria in asymmetric 2-player games with the Nash equilibria that can be found in the corresponding counterpart games. We also examine the stability properties of the corresponding rest points of the replicator dynamics in these games. Then we experimentally illustrate these findings in some canonical examples.

%{\color{blue} Julien: To Karl. Do we have simple plots of symmetric RD showing ESS, just Nash and rest points that are neither ESS nor Nash.}
\subsection*{Theoretical Findings}\label{sec:findings}

In this section, we prove the following result: if $(x,y) \in NE(A,B)$ (where $x$ and $y$ have the same support), then $x \in NE(B^\top)$ and $y \in NE(A)$. In addition, we prove that the reverse is true: if $x \in NE(B^\top)$ and $y \in NE(A)$ (where $x$ and $y$ have the same support) then $(x,y) \in NE(A,B)$. We will prove this result in two steps (Theorem~\ref{the:one} and its generalization Theorem~\ref{the:two}).

The theorems introduced apply to games where both players can play the same number of actions (i.e. square games). This condition can be weakened by adding dominated strategies to the player having the smallest number of actions (see the extended Battle of the Sexes example in the experimental section). Thus, without loss of generality, the theory will focus on square games. To begin, we state an important well-known property of Nash equilibria, that has been given different names; Gintis calls it fundamental theorem of Nash equilibria~\cite{Gintis09}. For sake of completeness, we provide a proof.
\newpage
\begin{property}
\label{TheoreticalResults:prop:Nash}
% Strategies $x$ and $y$ constitute a Nash equilibrium of an asymmetric normal form game $(A,B)$ if and only if
% \begin{align}
% &\forall i, (A y)_i \leq x^\top A y, \textrm{ and } \forall i \in I_x, (A y)_i = x^\top A y \\
% & \forall i, (x^\top B)_i \leq x^\top B y, \textrm{ and } \forall i \in I_y, (x^\top B)_i = x^\top B y
% \end{align}
Let the strategy profile $(x,y)$ be a Nash equilibrium of an asymmetric normal form
game $(A,B)$, and denote $I_z = \{i \; | \; z_i>0\}$ the support of a strategy $z$.
Then, 
\begin{align}
& z^\top A y = x^\top A y ~\textrm{ for all } z \textrm{ such that } I_z \subset I_x,   ~\textrm{ and }\\
& x^\top B z = x^\top B y ~\textrm{ for all } z \textrm{ such that } I_z \subset I_y. 
\end{align}
% where $I_x = \{i \; | \; x_i>0\}$ and $I_y = \{i \; | \; y_i>0\}.$
\end{property}

\begin{proof}
This result is widely known. We provide it as it is a basis of our theoretical results and for the sake of completeness.

If $x$ and $y$ constitute a Nash equilibrium then, by definition 
$z^\top A y \leq x^\top A y, \forall z$. Let us suppose that there exists a 
$z$ with $I_z \subset I_x$ such that $z^\top A y < x^\top A y$. 
Then there is a $i \in I_z \subset I_x$ satisfying $(A y)_i < x^\top A y$,
and we get $x^\top A y = \sum \limits_{i \in I_x} x_i (A y)_i < \sum \limits_{i \in I_x} x_i x^\top A y = x^\top A y$, which is a contradiction, proving the first claim. The claim for $B$ follows analogously.
% If $x$ and $y$ constitute a Nash equilibrium then, by definition, $\forall i, (A y)_i \leq x^\top A y$ and $\forall i, (x^\top B)_i \leq x^\top B y$. Let us suppose that there exists an $\hat{i} \in I_x$ or an $\tilde{i} \in I_y$ such that $(A y)_{\hat{i}} < x^\top A y$ or $(x^\top B)_{\tilde{i}} < x^\top B y$ then, we would have $\sum \limits_{\hat{i} \in I_x} x_{\hat{i}} (A y)_{\hat{i}} < \sum \limits_{\hat{i} \in I_x} x_{\hat{i}} x^\top A y$ or $\sum \limits_{\tilde{i} \in I_y} y_{\tilde{i}} (x^\top B)_{\tilde{i}} < \sum \limits_{\tilde{i} \in I_y} y_{\tilde{i}} x^\top B y$. Finally we would have $x^\top A y < x^\top A y$ or $x^\top B y < x^\top B y$ which is a contradiction and proves the first implication. The other implication is obvious.
%\qed
\end{proof}

\begin{property}
\label{TheoreticalResults:prop:Nash:onepop}
Let the strategy $x$ be a Nash equilibrium of a single population game $A$.
Then, 
\begin{align}
& z^\top A x = x^\top A x ~\textrm{ for all } z \textrm{ such that } I_z \subset I_x. 
\end{align}
\end{property}
\begin{proof}
The proof is similar to the proof of Property \ref{TheoreticalResults:prop:Nash}.
\end{proof}

This property will be useful in the steps of the proofs that follow.
We now present our first main result: a correspondence between the Nash equilibria of full support in the asymmetric game with those of full support in the counterpart games. Theorem~\ref{the:two} subsumes this result and we introduce this simpler version first for the sake of readability.

\begin{theorem}\label{the:one}
If strategies $x$ and $y$ constitute a Nash equilibrium of an asymmetric normal form game $G=(S_1,S_2,A,B)$, with both $x_i>0$ and $y_j>0$ for all $i,j$ (full support), and $|S_1|=|S_2|=n$, then it holds that $x$ is a Nash equilibrium of the single population game $B^T$ and $y$ is a Nash equilibrium of the single population game $A$. The reverse is also true.
%\begin{equation*}
% \frac{dx_i}{dt}=x_i((Ay)_i - x^TAy)=0, \qquad 
% \frac{dy_i}{dt}=y_i((x^TB)_i-x^TBy)=0
%\end{equation*}
%if and only if,
%\begin{equation*}
% \frac{dy_i}{dt}=y_i((Ay)_i - y^TAy)=0, \qquad
% \frac{dx_i}{dt}=x_i((x^TB)_i-x^TBx)=0
%\end{equation*}
%holds for the symmetric counterparts equations of the asymmetric NFG $G$. 
\end{theorem}

\begin{proof}
%%%%%%%%%%%%%%%%%%%%%%%%%%%%%%%%%%%%%%%%%%%%%%%%%%%%%%%%%%%%%%%%%%%
This result follows naturally from Property~\ref{TheoreticalResults:prop:Nash} and is implied by Theorem~\ref{the:two}. %This result is presented first for the sake of readability.

We start by assuming that $x$ and $y$ constitute a full support Nash equilibrium of the asymmetric game $(A,B)$. By Property~\ref{TheoreticalResults:prop:Nash} and since $x$ and $y$ have full support, we know that:
\begin{equation*}
    Ay = (1,...,1){^T}\ \max_{i \in \{1,...,n\}}\ (Ay)_i \quad \text{and,} \quad  x^TB = (1,...,1)\ \max_{i \in \{1,...,n\}}\ (x^TB)_i
\end{equation*}
From this we also know that $y^TAy$ = $(Ay)_i$ (since the $(Ay)_i$ are equal for all $i$'s in the vector $Ay$, so multiplying $Ay$ with $y^T$ will yield the same number $\max_i\ (Ay)_i$), and similarly $(x^TB)_i$ = $x^TBx$ (and thus $(B^T x)_i$ = $x^TB^Tx$), implying that:
\begin{equation*}
    \forall y', y^T A y = y'^T A y \quad \text{and,} \quad \forall x', x^T B^T x = x'^T B^T x
\end{equation*}

This concludes the proof.
%%%%%%%%%%%%%%%%%%%%%%%%%%%%%%%%%%%%%%%%%%%%%%%%%%%%%%%%%%%%%%%%%%%
% \begin{equation*}
%  y_i((Ay)_i - y^TAy)=0 \qquad \text{and} \qquad x_i((x^TB)_i-x^TBx)=0 
% \end{equation*}
% %
% We still have to prove the other direction. We assume that we have $x, y > 0$ and the above symmetrised equations, which means that
% \begin{equation*}
% Ay = (1,...,1)^T  \max(Ay) \qquad \text{and} \qquad x^TB=(1,...,1)\max(x^TB)
% \end{equation*}
% %
% and that,
%     $\forall i \quad (Ay)_i = y^TAy$,
% which implies in turn that $y^TAy=\max \limits_{i \in \{1,...,n\}} (Ay)_i$ and thus $\forall i \; (Ay)_i=\max \limits_{i \in \{1,...,n\}} (Ay)_i$
% %
% Now we know that for any $x>0$ : $x^TAy=\max \limits_{i \in \{1,...,n\}} (Ay)_i$, and thus
% \begin{equation*}
%  x_i((Ay)_i - x^TAy)=0
% \end{equation*}
% must hold. We can apply the same reasoning for $x^TB$, leading to,
% %
% \begin{equation*}
% y_i((x^TB)_i-x^TBy)=0    
% \end{equation*}
% %
% %\qed
\end{proof}

%Both players consider the same distribution of their respective strategies in the symmetrised counterpart replicator equations.
For the first counterpart game this means that the players will use the \emph{$y$} part of the Nash equilibrium of player 2 of the original asymmetric game, in the symmetric counterpart game determined by payoff table $A$. And similarly, for the second counterpart game this means that players will play according to the \emph{$x$} part of the Nash equilibrium of player 1 of the original asymmetric game, in the symmetric game determined by payoff table $B$. As such both players consider a symmetric version of the asymmetric game, for which this $y$ component and $x$ component constitute a Nash equilibrium in the two new respective symmetric games.

In essence, these two symmetric counterpart games can be considered as a decomposition of the original asymmetric game, which gives us a means to illustrate in a smaller strategy space where the mixed and pure equilibria are located. 

A direct consequence of Theorem \ref{the:one} is the following corollary that gives insights on the geometrical structure of Nash equilibrium,
\begin{corollary}
Combinations of Nash equilibria of full support of the games corresponding to the symmetrical counterparts of the original asymmetric game also form Nash equilibria of full support in this asymmetric game.
\end{corollary}

\begin{proof}
This is a direct consequence of Theorem \ref{the:one}.
%\qed
\end{proof}

The next theorem explores the case where the equilibrium is not of full support. We prove that the theorem stands if the strategies of both players have the same support. Indeed, the first theorem requires that both players play all actions with a positive probability, here we will only require that they play the actions with the same index with a positive probability. We say that $x$ and $y$ have the same support if the set of played actions $I_x = \{i \; | \; x_i>0\}$ and $I_y = \{i \; | \; y_i>0\}$ are equal.

\begin{theorem}\label{the:two}
Strategies $x$ and $y$ constitute a Nash equilibrium of an asymmetric game $G = (S_1, S_2, A,B)$ with the same support (i.e. $I_x=I_y$) if and only if $x$ is a Nash equilibrium of the single population game $B^T$, $y$ is a Nash equilibrium of the single population game $A$ and $I_x = I_y$.

%\begin{minipage}[b]{.3\textwidth}
%\vspace{-\baselineskip}
%\begin{equation}
% \frac{dx_i}{dt}=x_i((Ay)_i - x^TAy) \label{ESSCoupledA}
%\end{equation}
%\end{minipage}%
%\qquad and
%\begin{minipage}[b]{.3\textwidth}
%\vspace{-\baselineskip}
%\begin{equation}
%\frac{dy_i}{dt}=y_i((x^TB)_i-x^TBy) \label{ESSCoupledB}
%\end{equation}
%\end{minipage}
%\end{center}
%
%\noindent \\ if and only if $y$ and $x$ are (strict) Nash equilibria of the respective symmetric counterpart games,\\
%
%\begin{minipage}[b]{.3\textwidth}
%\vspace{-\baselineskip}
%\begin{equation}
% \frac{dy_i}{dt}=y_i((Ay)_i - y^TAy) \label{ESSDecoupledA}
%\end{equation}
%\end{minipage}%
%\qquad and
%\begin{minipage}[b]{.3\textwidth}
%\vspace{-\baselineskip}
%\begin{equation}
%\frac{dx_i}{dt}=x_i((x^TB)_i-x^TBx) \label{ESSDecoupledB}
%\end{equation}
%\end{minipage}

\end{theorem}

\begin{proof}
We start by assuming that $x$ and $y$ constitute a Nash equilibrium of same support ($I_x = I_y$) of the asymmetric game $(A,B)$. By Property~\ref{TheoreticalResults:prop:Nash}, and since $x$ and $y$ have the same support, we know that:
\begin{align}
& z^\top A y = x^\top A y ~\textrm{ for all } z \textrm{ such that } I_z \subset I_x,   ~\textrm{ and }\\
& x^\top B z' = x^\top B y ~\textrm{ for all } z' \textrm{ such that } I_{z'} \subset I_y. 
\end{align}
implying that $y^\top A y = x^\top A y$ and $x^\top B x = x^\top B y$ (by setting $z=y$ and $z'=x$). Then, from the Nash equilibrium condition we can write:\\

$\forall x' \in \Delta S_1,\; x^T A y \geq x'^T A y$ and $\forall y' \in \Delta S_2,\;  x^T B y \geq x^T B y'$

$\forall y' \in \Delta S_2,\; y^\top A y \geq y'^T A y$ and $\forall x' \in \Delta S_1,\;  x^\top B x \geq x^T B x'$

$\forall y' \in \Delta S_2,\; y^\top A y \geq y'^T A y$ and $\forall x' \in \Delta S_1,\;  x^\top B^\top x \geq x'^T B^\top x$
\\
\\which implies that $y$ is a Nash equilibrium of $B^\top$ and $x$ is a Nash equilibrium of $A$.\newline

The proof of the other direction follows similar mechanics and uses Property~\ref{TheoreticalResults:prop:Nash:onepop}. Let us now assume that $y$ is a Nash equilibrium of $B^\top$ and $x$ is a Nash equilibrium of $A$ with $I_x = I_y$. Then, from Property~\ref{TheoreticalResults:prop:Nash:onepop} we have:
\begin{align}
& z^\top A y = y^\top A y ~\textrm{ for all } z \textrm{ such that } I_z \subset I_y,   ~\textrm{ and }\\
& z'^\top B^\top x = x^\top B^\top x ~\textrm{ for all } z' \textrm{ such that } I_{z'} \subset I_x. 
\end{align}
In particular we get $y^\top A y = x^\top A y$ and $x^\top B x = x^\top B y$ (by setting $z=x$ and $z'=y$). From the Nash equilibrium condition of the single population games we can write:\\

$\forall y' \in \Delta S_2,\; y^\top A y \geq y'^T A y$ and $\forall x' \in \Delta S_1,\;  x^\top B^\top x \geq x'^T B^\top x$

$\forall y' \in \Delta S_2,\; y^\top A y \geq y'^T A y$ and $\forall x' \in \Delta S_1,\;  x^\top B x \geq x^T B x'$

$\forall x' \in \Delta S_1,\; x^T A y \geq x'^T A y$ and $\forall y' \in \Delta S_2,\;  x^T B y \geq x^T B y'$
\\
\\which concludes the proof.

\end{proof}

\begin{corollary}\label{the:three}
Strategies $x$ and $y$ constitute a pure (strict) Nash equilibrium of an asymmetric normal form game $G = (S_1, S_2, A,B)$, with support on the strategy with the same index in their respective strategy sets $S_1$ and $S_2$, if and only if, 
$y$ and $x$ are also pure (strict) Nash equilibria of the counterpart games defined by A,
\begin{equation}
 \frac{dy_i}{dt}=y_i((Ay)_i - y^TAy)=0 \label{ESSDecoupledAth3}
\end{equation}
and B,
\begin{equation}
\frac{dx_i}{dt}=x_i((x^TB)_i-x^TBx)=0 \label{ESSDecoupledBth3}
\end{equation}
\end{corollary}

\begin{proof}
This is a direct consequence of Theorem~\ref{the:two}.
%\qed
\end{proof}

The theorems can only be used for equilibria in the counterpart games with matching supports ($I_x = I_y$) from both players. One can work around this condition though by simply permuting the actions of one player in matrix $A$ and $B$ to study all configurations of supports of the same cardinality. To be precise, we need to analyze all the counterpart games defined by $A_\Sigma = A \Sigma$ and $B_\Sigma^T = (B \Sigma)^T$ for all permutation matrices $\Sigma$.
This technique is sufficient to study non-degenerate games, as in a non-degenerate game all Nash equilibria have a support of same size (in a non-degenerate game if $(x,y)$ is a Nash equilibrium then $|I_x|=|I_y|$\cite{VonStengel20021723}).

\subsubsection*{Stability Analysis}
We can now examine the stability of the pure Nash equilibria discussed in the previously derived theorems.

\begin{corollary}
Strategy $y$ is a strict Nash equilibrium of the first counterpart game defined by A and strategy $x$ is a strict Nash equilibrium of the second counterpart game defined by B, if and only if,\\ $(x,y)$ is a locally asymptotically stable equilibrium and a two-species ESS of the asymmetric normal form game $G = (S_1, S_2, A, B)$ with support on the strategy with the same index in their respective strategy sets $S_1$ and $S_2$.
\end{corollary}

\begin{proof}
This a direct consequence of Corollary \ref{the:three}. More specifically, from Corollary \ref{the:three} we know that $(x,y)$ is a strict Nash equilibrium of $G$.
It has been shown that $(x,y)$ is a strict Nash equilibrium of $G$ iff it is a two-species ESS~\cite{Hofbauer98,Selten80,Cressman14}. 
\end{proof}

%\section*{Evolutionary Analysis of Asymmetric Games}
\subsection*{Experimental illustration}
We will now illustrate how the theoretical links between asymmetric games and their counterpart symmetric replicator dynamics facilitate analysis of asymmetric multiagent games, and provide a convenient tool to get insight into their equilibrium landscape. We do this for several examples. The first example concerns the Battle of the Sexes game to illustrate the intuition behind the results. The second example extends the Battle of the Sexes game with one strategy for one of the players, illustrating the permutation argument of the theorems and how to apply the results in case of a non-square game. The third example is a bimatrix game generated in the context of a multiagent learning algorithm called PSRO (Policy Space Response Oracles~\cite{Lanctot17}) and concerns Leduc Poker. This algorithm produces normal-form ``empirical games'' which each correspond to an extensive-form game with a reduced strategy space, using incremental best response learning. Finally, the last asymmetric game illustrates the theorems for a single mixed equilibrium of full support, while its counterpart games have many more equilibria.

%There is a fundamental complexity that arises when using the evolutionary dynamics of an (2-player) asymmetric game to analyse its equilibrium structure, i.e., when player 1 changes its strategy this causes a change in the dynamics of player 2, and vice versa, and as such we can no longer show the static plots according to which the dynamics evolve for each player, but we would rather have to look at a movie, showing how the dynamics change for both players simultaneously, when the other players change their strategy (we will illustrate this in the PSRO-produced game on Leduc Poker). 

A fundamental complexity arises when using the evolutionary dynamics of a 2-player asymmetric game to analyse its equilibrium structure, as the dynamics for the two players is intrinsically coupled and high-dimensional. While one could fix a player's strategy and consider the induced dynamics for the other player in its respective strategy simplex, a static trajectory plot of this would not faithfully represent the complexity of the full 2-player dynamics. To gain a somewhat more complete intuitive picture, one can represent this dynamics as a movie, showing the change in induced dynamics for one player, as one varies the (fixed) strategy for the other (we will illustrate this in the PSRO-produced game on Leduc Poker).

The theorems introduced in the previous section help to overcome this problem, and allow to analyse the evolutionary dynamics of the symmetric counterpart games instead of the asymmetric game itself, revealing the landscape of Nash equilibria, which seriously simplifies the analysis.

%\textbf{We'd like to illustrate the following: the intuition in a simple $2 \times 2$ asymmetric game, a non-squared game, a game with several equilibria in the counterpart games, and the application to PSRO?}

\subsubsection*{Battle of the Sexes}

Symmetry assumes that strategy sets and corresponding payoffs are the same for all players in the interaction. An example of an asymmetric game is the Battle of the Sexes (BoS) game illustrated in Table \ref{fig:BoS}. In this game both players do have the same strategy set, i.e., go to the \textit{opera} or go to the \textit{movies}, however, the corresponding payoffs for each are different, expressing the difference in preferences that both players have over their choices.

\noindent The Battle of the Sexes has two pure Nash equilibria, which are ESS as well (located at coordinates $(0,0)$ and $(1,1)$), and one unstable completely mixed Nash equilibrium in which the players play respectively $x=(\frac{3}{5},\frac{2}{5})$ and $y=(\frac{2}{5},\frac{3}{5})$. Figure \ref{fig:dfieldBoS} illustrates the two-player evolutionary dynamics using the replicator equations, in which the x-axis corresponds to the probability with which player 1 plays $O$ (Opera), and the y-axis corresponds to the probability with which the 2nd player plays $O$ (Opera). The blue arrows show the vector field and the black lines are the corresponding trajectories. Note that it is still possible here to capture all of the dynamics in a static plot for the case of 2-player 2-action games, but is generally not possible in games with more than two actions.

\begin{figure}[h!]
    \centering
    \includegraphics[width=8cm]{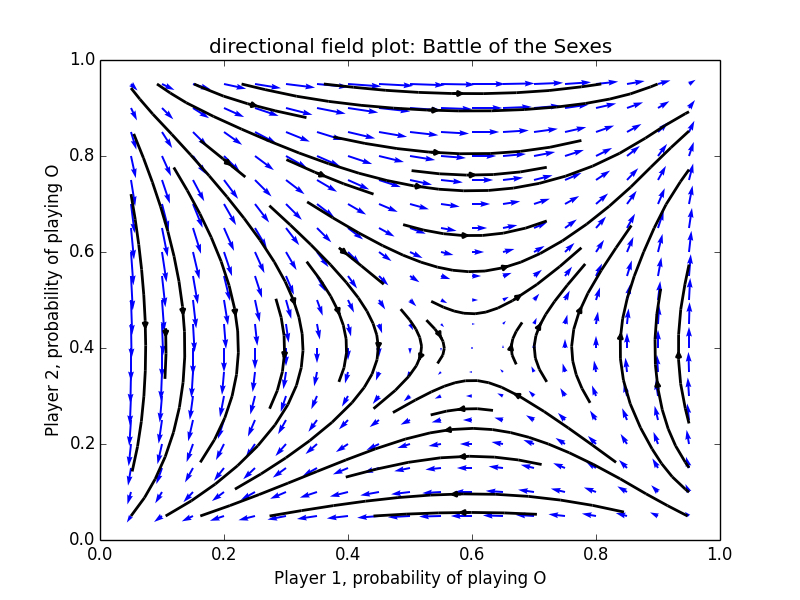}
    \caption{Legend. Directional field plot of the Battle of the Sexes game.}
    \label{fig:dfieldBoS}
\end{figure}

We now use this game to illustrate Theorem \ref{the:one}. 
If we apply Theorem \ref{the:one} we know that the first and second counterpart symmetric games can be described by the payoff tables shown in Table \ref{fig:BoScounterpart}. The first counterpart game has $((\frac{2}{5},\frac{3}{5}),(\frac{2}{5},\frac{3}{5}))$ as a mixed Nash equilibrium, and the second counterpart game has $((\frac{3}{5},\frac{2}{5}),(\frac{3}{5},\frac{2}{5}))$ as a mixed Nash equilibrium.

\begin{table}[h!]
	\centering
\begin{minipage}{.2\textwidth}
   \begin{game}{2}{2}[][]
   	    &  O     &  M     \\
   	 O  &    $3$      & $0$  \\
   	 M &  $0$ & $2$\\
   \end{game}
\end{minipage}% `
\begin{minipage}{.2\textwidth}
   \begin{game}{2}{2}[][]
   	    &  O      &  M     \\
   	 O  &    $2$      & $0$  \\
   	 M &  $0$ & $3$\\
   \end{game}
\end{minipage}
\caption{Counterpart matrix game 1 and 2 for the Battle of the Sexes game.}
\label{fig:BoScounterpart}
\end{table}

\begin{figure}[ht]
    \centering
    \includegraphics[width=13cm]{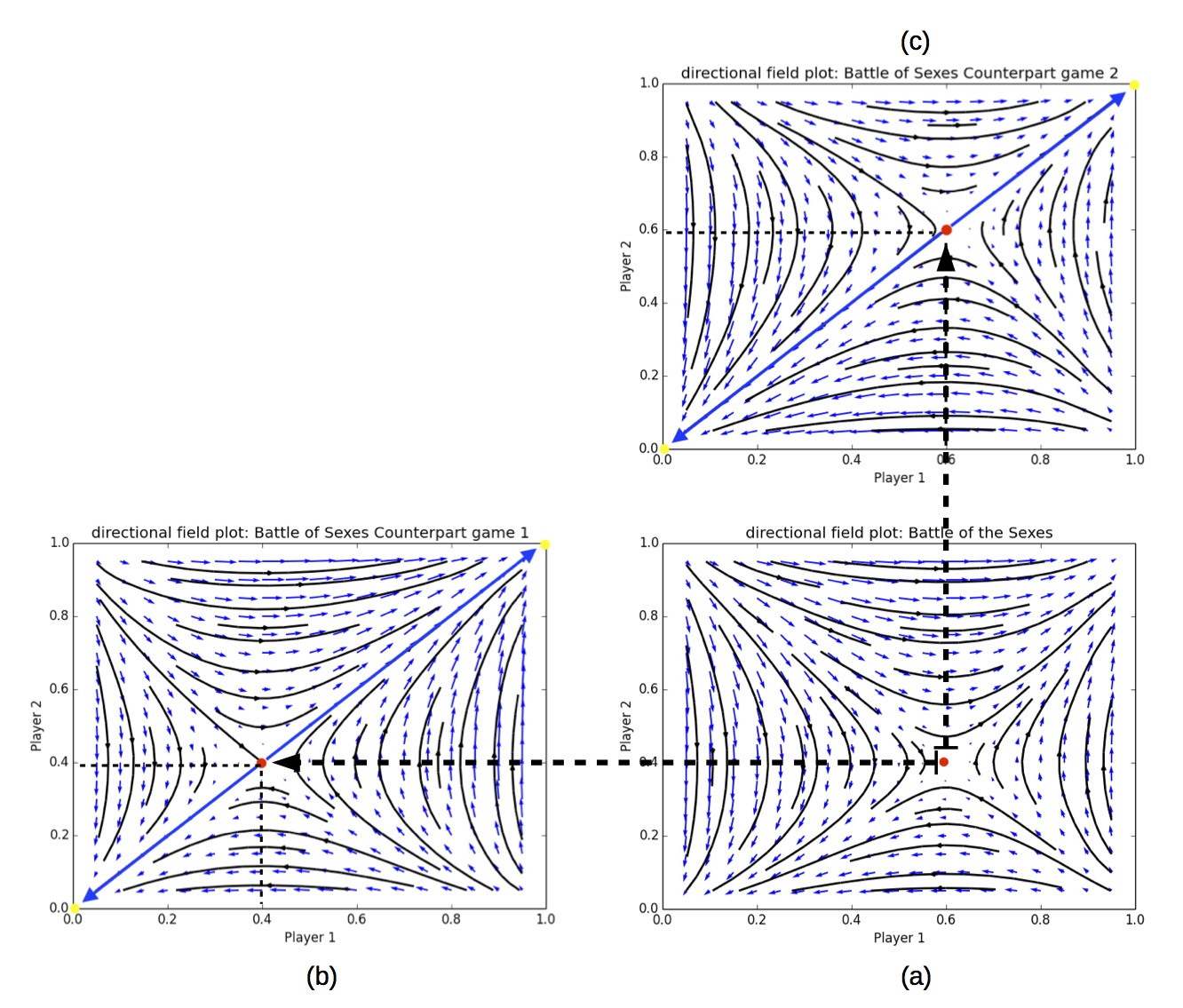}
    \caption{This plot shows a visual representation of how the mixed Nash equilibrium is decomposed into Nash equilibria in both counterpart games. (a) shows the directional field plot of the Battle of the Sexes game. (b) illustrates how the y-component of the asymmetric Nash equilibrium becomes a Nash equilibrium in the first counterpart game, and (c) shows how the x-component of the asymmetric Nash equilibrium becomes a Nash equilibrium in the first counterpart game.}
    \label{fig:dfieldBoS-composed}
\end{figure}

 In Figures \ref{fig:dfieldBoS-composed}(b) and \ref{fig:dfieldBoS-composed}(c) we show the evolutionary dynamics of both counterpart games, from which the respective equilibria can be observed, as predicted by Theorem \ref{the:one}. 

 Additionally, we also know that the reverse holds, i.e., if we were given the symmetric counterpart games, we would know that $((\frac{3}{5},\frac{2}{5}),(\frac{2}{5},\frac{3}{5}))$ would also be a mixed Nash equilibrium of the original asymmetric BoS. In this case we can combine the mixed Nash equilibria of both counterpart games into the mixed Nash equilibrium of the original asymmetric game, as prescribed by Theorem \ref{the:one}. Specifically, as $y=(\frac{2}{5},\frac{3}{5})$ is part of the Nash equilibrium in the first counterpart game and $x=(\frac{3}{5},\frac{2}{5})$ in the second counterpart game, we can combine them into $(x=(\frac{3}{5},\frac{2}{5}),y=(\frac{2}{5},\frac{3}{5}))$, which is a mixed Nash equilibrium of full support of the asymmetric Battle of the Sexes game.

 If we now apply Theorem \ref{the:two} to the Battle of the Sexes game, then we find that pure strategy Nash equilibria $x=(1,0)$ (and $y=(1,0)$ for the second counterpart) and $x=(0,1)$ (and $y=(0,1)$ for the second counterpart), which are both ESS, are also Nash equilibria in the counterpart games shown in Tables \ref{fig:BoScounterpart}. Also here the reverse holds, i.e., if we know the counterpart games, and we observe that $x=(1,0)$ and $x=(0,1)$ ($y=(1,0)$ and $y=(0,1)$ for the other counterpart of the game) are Nash in both games, we know that $x=y=(1,0)$ and $x=y=(0,1)$ are also Nash in the original asymmetric game. This can also be observed in Figures \ref{fig:dfieldBoS-composed}(b), \ref{fig:dfieldBoS-composed}(c) and \ref{fig:dfieldBoS-composed}(a). Specifically, the pure Nash equilibria are situated at coordinates $(0,0)$ and $(1,1)$ in Figures \ref{fig:dfieldBoS-composed}(b) and \ref{fig:dfieldBoS-composed}(c). Furthermore, it is important to understand that the counterpart dynamics are visualised only on the diagonal from coordinates $(0,0)$ to $(1,1)$, as that is where both players play with the same strategy distribution over their respective actions.

\subsubsection*{Extended Battle of the Sexes game}
\label{BoS:Ambiguity}
In order to illustrate the theorems in a game that is non-square, including permutation of strategies, we extend the Battle of the Sexes game with a third strategy.
Specifically, we give the second player a third strategy $R$ in which she can choose to listen to a concert on the radio instead of going to the opera or movies with her partner. %The value of $a$ typically varies between $60-260$. 
This game is illustrated in Table \ref{fig:BoS-extended}.

\begin{table}[h!]
	\centering
   \begin{game}{2}{3}[][]
   	    &  O     &  R  & M   \\
   	 O  &    $3,2$      & $0.5,0.5$  & $0,0$\\
   	 M &  $0,0$ & $0.5,0.55$ & $2,3$\\
   \end{game}
\caption{Extended Battle of the Sexes game.}
\label{fig:BoS-extended}
\end{table}

 %If we would like to carry out a similar evolutionary analysis as before we need two populations for the asymmetric replicator equations. Note that now in this case the strategy sets of both players are different. Using the asymmetric replicator dynamics to plot the evolutionary dynamics quickly becomes complicated since we have to deal with two simplices, and when we consider a strategy for the first player in its corresponding simplex, that player changing its strategy will cause the second simplex to change, and vice versa. Consequently, it is not straightforward anymore to analyse the equilibrium landscape for both players, as any trajectory change in one simplex causes a change in the other simplex as well (a movie of the dynamics of the Leduc Poker empirical game will illustrate this). Again we can apply the counterpart RD theorems here to remedy this problem and consequently analyse the equilibrium structure in the symmetric counterpart games instead, yielding insight into the equilibrium landscape of the asymmetric game.

If we would like to carry out a similar evolutionary analysis as before we need two populations for the asymmetric replicator equations. Note that in this case the strategy sets of both players are different. Using the asymmetric replicator dynamics to plot the evolutionary dynamics quickly becomes complicated since the full dynamical picture is high-dimensional and not faithfully represented by projections to the respective player's individual strategy simplices. In other words, a static plot of the dynamics for one player does not immediately allow conclusions about equilibria, as it only describes a player's strategy evolution assuming a fixed (rather than dynamically evolving) strategy of the other player. Again we can apply the counterpart RD theorems here to remedy this problem and consequently analyse the equilibrium structure in the symmetric counterpart games instead, yielding insight into the equilibrium landscape of the asymmetric game.

 %We examine this game for $a=1.2$. 
 In Tables \ref{fig:CP-extBoS1} and \ref{fig:CP-extBoS2} we show the counterpart games A and B. Note that we introduce a \emph{dummy} action $D$ for the first player, in order to make sure that both players have the same number of actions in their strategy set (a requirement to apply the theorems) by just adding $-1$ for both players playing this strategy, which makes $D$ completely dominated and thus redundant.

\begin{table}[h!]
	\centering
\begin{minipage}[c]{.48\textwidth}
   \begin{game}{3}{3}[][]
   	    &  O     &  R   & M  \\
   	 O  &    $3$      & $0.5$ & $0$\\
   	 M &  $0$ & $0.5$ & $2$\\
   	 D & $-1$ & $-1$ & $-1$\\
   \end{game}
   \caption{Payoff matrix for the 1st counterpart game of the Extended BoS game. Strategy $D$ is added to make the matrix completely square.}
   \label{fig:CP-extBoS1}
\end{minipage}%
\hfill
\begin{minipage}[c]{.48\textwidth}
   \begin{game}{3}{3}[][]
   	  &  O     &  R   & M  \\
   	 O  &    $2$      & $0.5$ & $0$\\
   	 M &  $0$ & $0.55$ & $3$\\
   	 D & $-1$ & $-1$ & $-1$\\
   \end{game}
   \caption{Payoff matrix for the 2nd counterpart game of the Extended BoS game. Strategy $D$ is added to make the matrix completely square.}
   \label{fig:CP-extBoS2}
\end{minipage}
\end{table}

The three Nash equilibria of interest of this asymmetric game are the following, $\{(x=(0.6,0.4,0),y=(0.4,0,0.6)),(x=(0,1,0),y=(0,0,1)),(x=(1,0,0),y=(1,0,0)))\}$ (we use the online banach solver \url{http://banach.lse.ac.uk/} to check that the Nash equilibria we find are correct~\cite{Avis10}). 

We now look for the $y$ and $x$ parts of these equilibria in the counterpart games.
In Figure \ref{fig:dfieldextBoSCP1} we show the evolutionary dynamics of the first counterpart game and in Figure \ref{fig:dfieldextBoSCP2} the evolutionary dynamics of the second counterpart game. In the first counterpart we only need to consider the 1-face formed by strategies $O$ and $M$ as the third strategy is our dummy strategy. In this game there are two Nash equilibria, i.e., $(1,0,0)$ (stable, yellow oval) and $(0,1,0)$ (unstable, orange oval), so either playing $O$ or $M$. The second counterpart game also has two Nash equilibria, i.e., $(1,0,0)$ and $(0,0,1)$ playing either $O$ or $M$ as well. Note there are also two rest points at the faces formed by $O$ and $R$ and $O$ and $M$, which are not Nash (see Figure \ref{fig:rest_point} for an explanation). There is no mixed equilibrium of full support, so we cannot apply Theorem \ref{the:one} here. If we apply Theorem \ref{the:two} we know that $((1,0,0),(1,0,0))$ must also be a pure Nash equilibrium in the original asymmetric game, and we can remove the dummy strategy for player 1.
At this stage we are left with equilibria $(x=(0.6,0.4,0),y=(0.4,0,0.6))$ and $(x=(0,1,0),y=(0,0,1))$ in the asymmetric game for which we did not find a symmetric counterpart at this stage. Now the permutation of the counterpart games, explained earlier in the findings section, comes into play. Recall that in order to study all configurations of supports of the same cardinal for both players one needs to simply permute the actions of one player in matrix $A$ and $B$. Let's have a look at such a permutation, specifically, let's permute the 2nd and 3rd action for player 2, resulting in Tables \ref{fig:CP-extBoS1perm} and \ref{fig:CP-extBoS2perm}.

\begin{table}[h!]
	\centering
\begin{minipage}[c]{.48\textwidth}
   \begin{game}{3}{3}[][]
   	   &  O     &  M   & R  \\
   	 O  &    $3$      & $0$ & $0.5$\\
   	 M &  $0$ & $2$ & $0.5$\\
   	 D & $-1$ & $-1$ & $-1$\\
   \end{game}
   \caption{Permuted payoff matrix for the 1st counterpart game of the Extended BoS game.}
   \label{fig:CP-extBoS1perm}
\end{minipage}%
\hfill
\begin{minipage}[c]{.48\textwidth}
   \begin{game}{3}{3}[][]
   	  &  O     &  M   & R  \\
   	 O  &    $2$      & $0$ & $0.5$\\
   	 M &  $0$ & $3$ & $0.55$\\
   	 D & $-1$ & $-1$ & $-1$\\
   \end{game}
   \caption{Permuted payoff matrix for the 2nd counterpart game of the Extended BoS game.}
   \label{fig:CP-extBoS2perm}
\end{minipage}
\end{table}

Again we can analyse these counterpart games. Specifically, we find Nash equilibria $(1,0,0)$, $(0.4,0.6,0)$, and $(0,1,0)$ for permuted counterpart game 1 (Table \ref{fig:CP-extBoS1perm}), and Nash equilibria $(0,0,1)$, $(0.6,0.4,0)$, $(0,1,0)$, and $(1,0,0)$ for permuted counterpart game 2 (Table \ref{fig:CP-extBoS2perm}), which are illustrated in Figures \ref{fig:dfieldextBoSCP1perm} and \ref{fig:dfieldextBoSCP2perm}. From these identified Nash equilibria in both counterpart games we can combine the remaining Nash equilibria for the asymmetric game. Specifically, by applying Theorem \ref{the:two} we find $(x=(0.6,0.4,0),y=(0.4,0.6,0))$, which translates into $(x=(0.6,0.4,0),y=(0.4,0,0.6))$ for the asymmetric game as we permuted actions 2 and 3 for the second player and we need to swap these again. Additionally, we also find $(x=(0,1,0),y=(0,1,0))$, which translates into equilibrium $(x=(0,1,0),y=(0,0,1))$ for the asymmetric game as we permuted action 2 and 3 for the second player. Now we have found all Nash equilibria of the original asymmetric game. %Figures \ref{fig:dfieldextBoSCP1perm} and \ref{fig:dfieldextBoSCP2perm} illustrate the evolutionary dynamics of the permuted games.

So, also in this case, i.e., when the game is not square and strategies need to be permuted, the theorems are still applicable and allow for analysis of the original asymmetric game.

\begin{figure}[!tbp]
  \centering
  \begin{minipage}[b]{0.45\textwidth}
     \includegraphics[width=9cm]{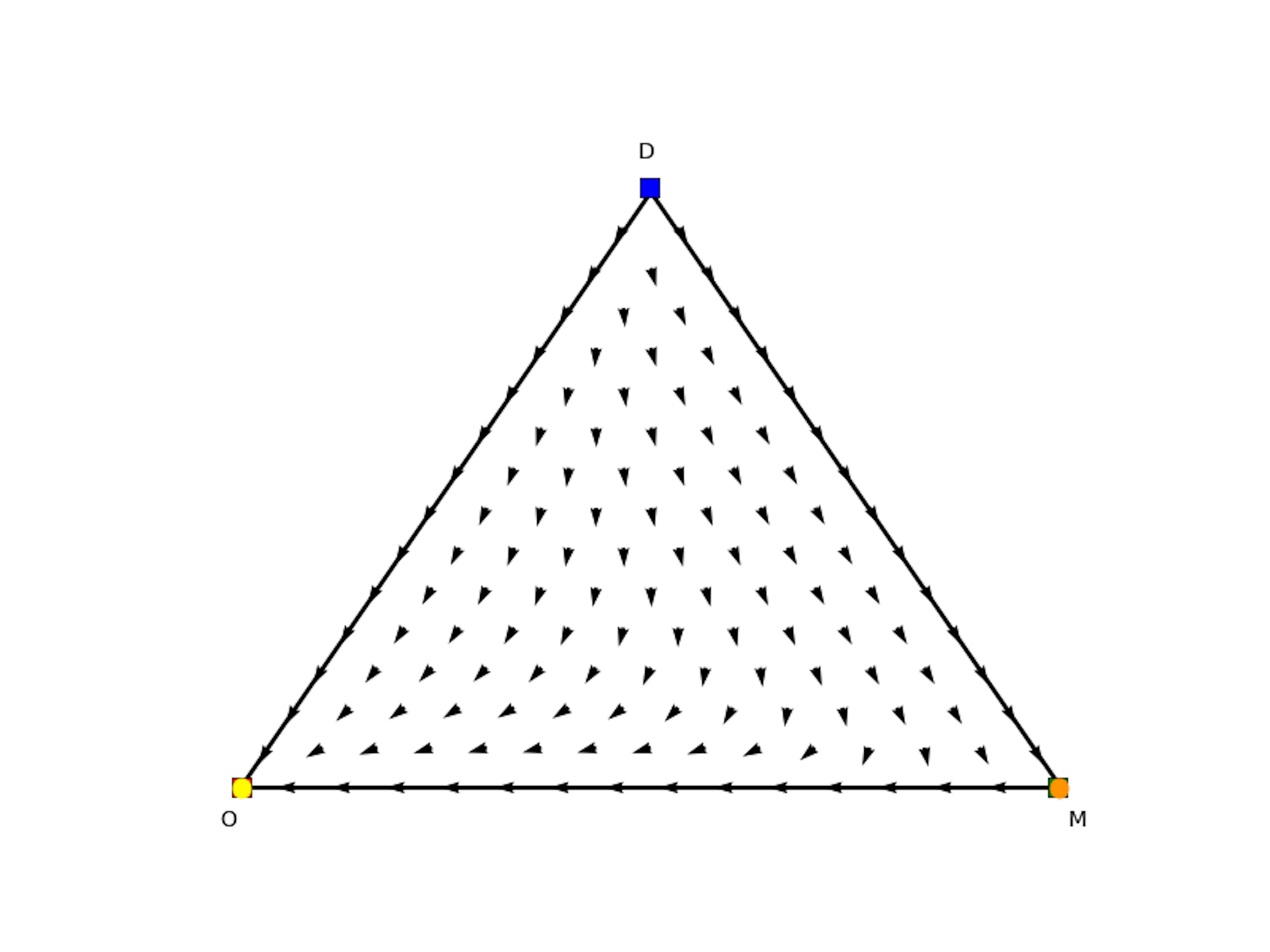}
    \caption{Legend. Directional field plot $\Sigma_3$ of the first counterpart game of the extended Battle of the Sexes game.}
    \label{fig:dfieldextBoSCP1}
  \end{minipage}
 % \hfill
 \qquad
  \begin{minipage}[b]{0.45\textwidth}
\includegraphics[width=9cm]{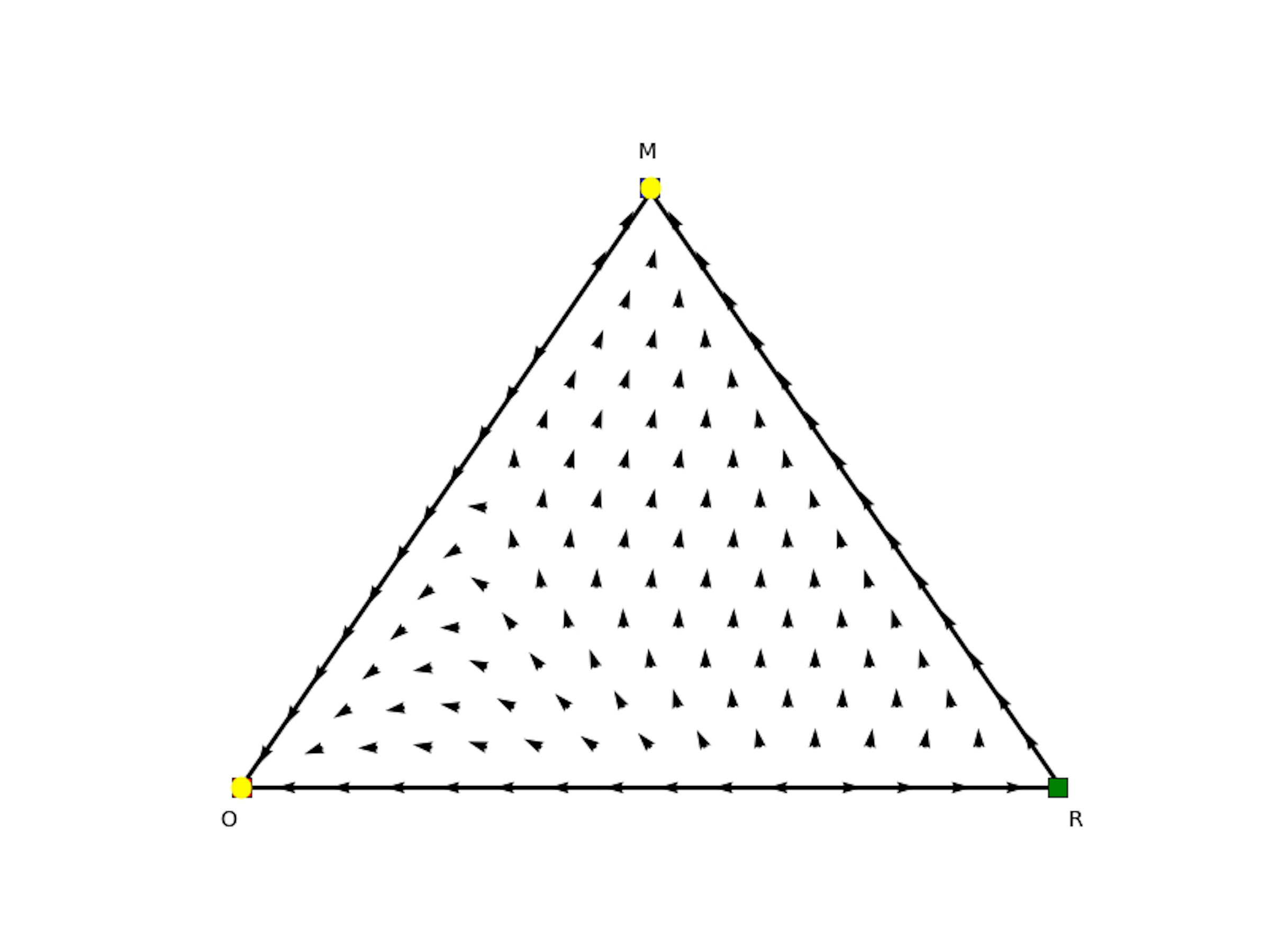}
    \caption{Legend. Directional field plot $\Sigma_3$ of the second counterpart game of the extended Battle of the Sexes game.}
    \label{fig:dfieldextBoSCP2}
  \end{minipage}
\end{figure}

\begin{figure}[!tbp]
  \centering
  \begin{minipage}[b]{0.45\textwidth}
     \includegraphics[width=9cm]{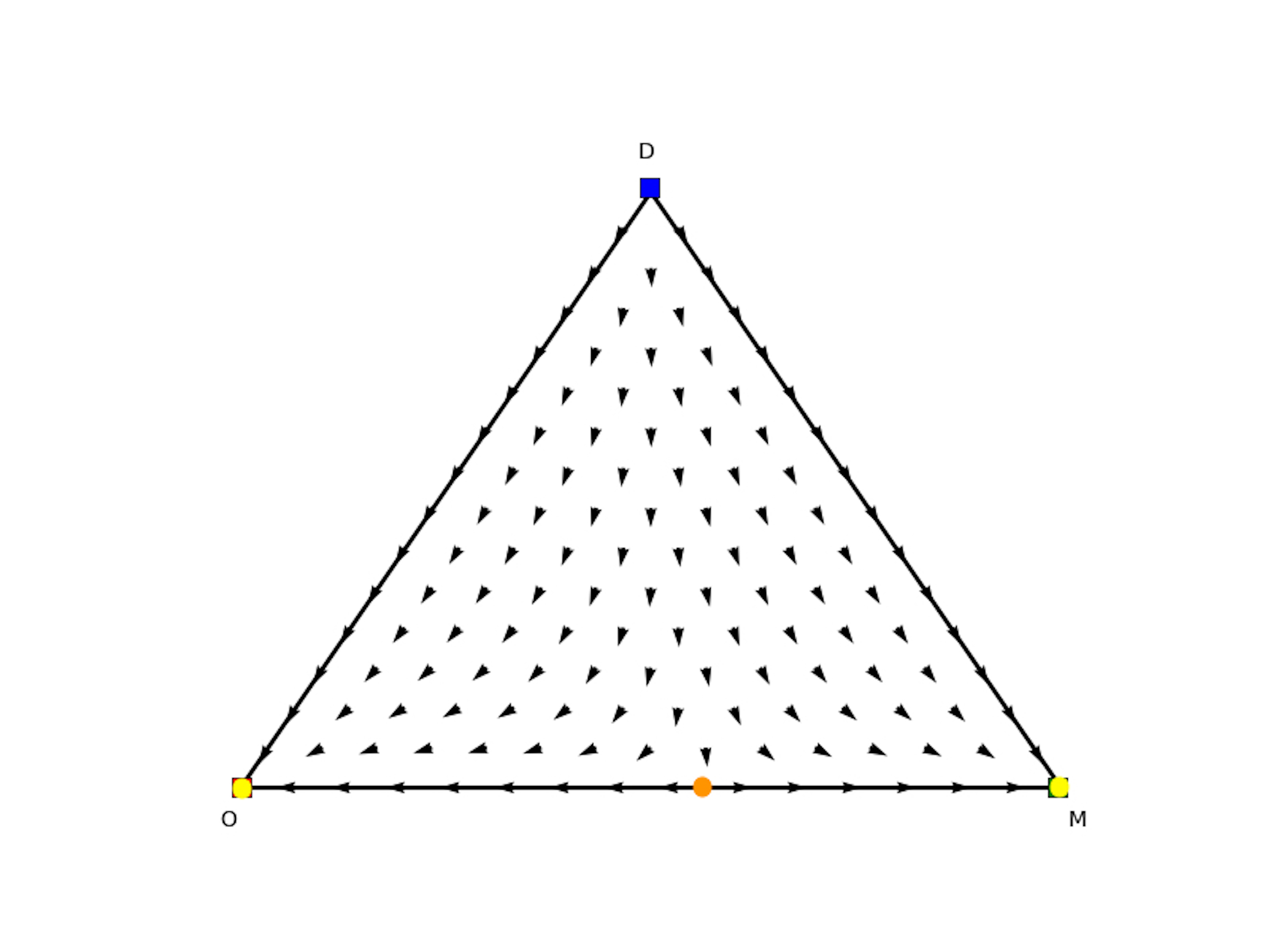}
    \caption{Legend. Directional field plot $\Sigma_3$ of the first counterpart game of the permuted extended Battle of the Sexes game.}
    \label{fig:dfieldextBoSCP1perm}
  \end{minipage}
 % \hfill
 \qquad
  \begin{minipage}[b]{0.45\textwidth}
\includegraphics[width=9cm]{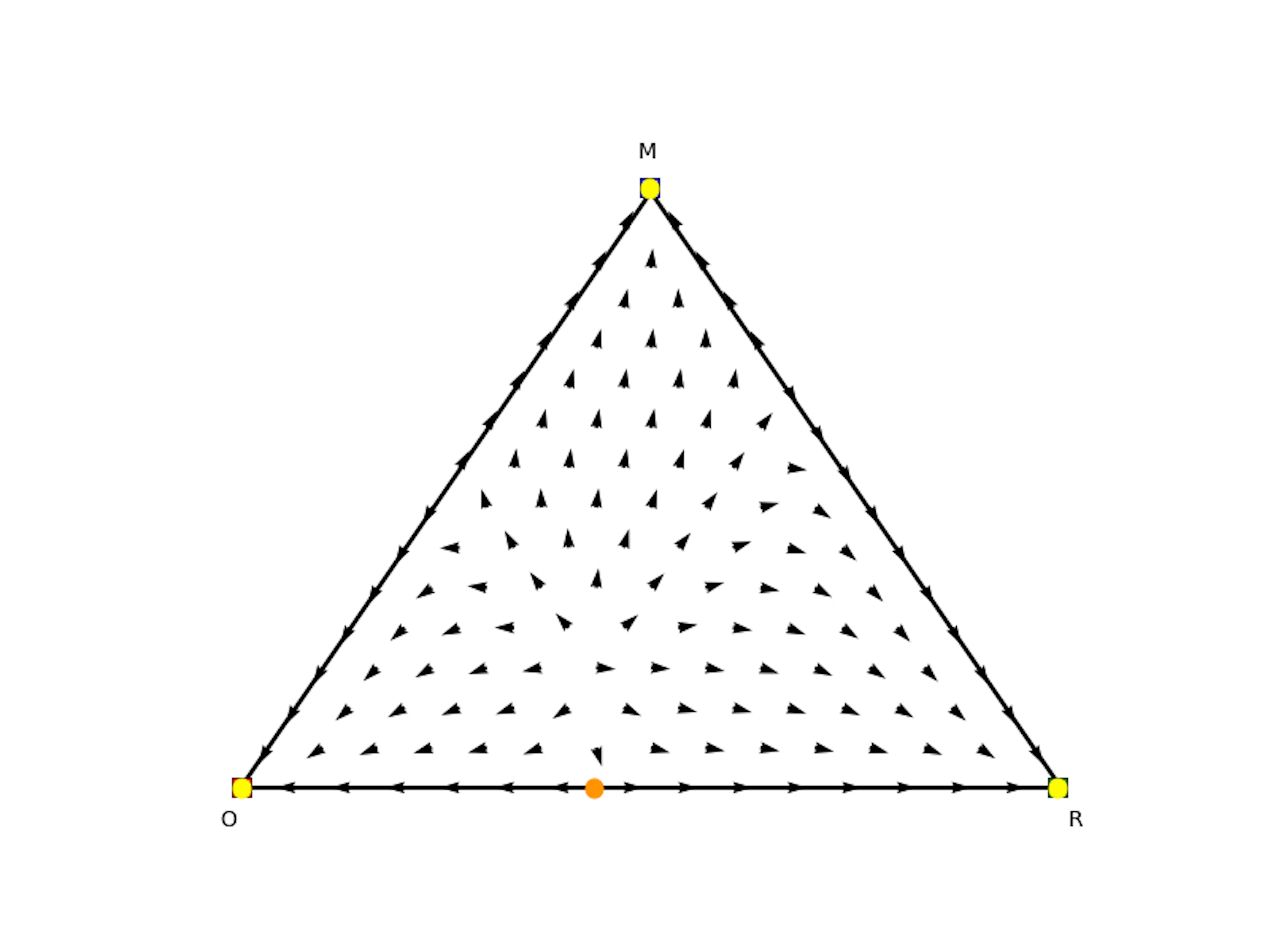}
    \caption{Legend. Directional field plot $\Sigma_3$ of the second counterpart game of the permuted extended Battle of the Sexes game.}
    \label{fig:dfieldextBoSCP2perm}
  \end{minipage}
\end{figure}

\subsubsection*{Poker generated asymmetric games}

Policy Space Response Oracles (PSRO) is a multiagent reinforcement learning process that reduces the strategy space of large extensive-form games via iterative best response computation. PSRO can be seen as a generalized form of fictitious play that produces approximate best responses, with arbitrary distributions over generated responses computed by meta-strategy solvers. 
PSRO was applied to a commonly-used benchmark problem in artificial intelligence research known as Leduc poker~\cite{Southey05}. Leduc poker has a deck of 6 cards (jack, queen, king in two suits). Each player receives an initial private card, can bet a fixed amount of 2 chips in the first round, 4 chips in the second round (with a maxium of two raises in each round). Before the second round starts, a public card is revealed. 

In Table \ref{fig:PSROgame} we present such an asymmetric $3 \times 3$ 2-player PSRO generated game, playing Leduc Poker. In the game illustrated here, each player has three strategies that, for ease of the exposition, we call $\{A, B, C\}$ for player 1, and $\{D, E, F\}$ for player 2. Each one of these strategies represents a larger strategy in the full extensive-form game of Leduc poker, specifically an approximate best response to a distribution over previous opponent strategies. The game produced here then is truly asymmetric, in the sense that the strategy spaces in the original game are inherently asymmetric since player 1 always starts each round, the strategy spaces are defined by different (mostly unique) betting sequences, and even under perfect equilibrium play there is a slight advantage to player 2~\cite{Lanctot17}. So, both players have significantly different strategy sets. In Tables \ref{fig:CP-PSRO1} and \ref{fig:CP-PSRO2} we show the two symmetric counterpart games of the empirical game produced by PSRO on Leduc poker.

% P1
%-0.160  0.050  0.080
%0.640 -0.170 -0.310
%0.790  1.060 -0.370

% P2
% 0.160 -0.050 -0.080
%-0.640  0.170  0.310
%-0.790 -1.060  0.370

\begin{table}[h!]
	\centering
   \begin{game}{3}{3}[][]
   	 & D & E & F  \\
    A & $-0.16,0.16$ & $0.05,-0.05$ & $0.08,-0.08$ \\
    B & $0.64,-0.64$ & $-0.17,0.17$ & $-0.31,0.31$  \\
    C & $0.79,-0.79$ & $1.06,-1.06$ & $-0.37,0.37$  \\
   \end{game}
\caption{Payoff matrix of an asymmetric empirical game produced by PSRO applied to Leduc poker.}
 \label{fig:PSROgame}
\end{table}

\begin{table}[h!]
	\centering
\begin{minipage}{.48\textwidth}
   \begin{game}{3}{3}[][]
   	    &  A     &  B   & C  \\
   	 A  &    $-0.16$      & $0.05$ & $0.08$\\
   	 B &  $0.64$ & $-0.17$ & $-0.31$\\
   	 C & $0.79$ & $1.06$ & $-0.37$\\
   \end{game}
   \caption{First counterpart game of the Leduc poker empirical game.}
   \label{fig:CP-PSRO1}
\end{minipage}%
\hfill
\begin{minipage}{.48\textwidth}
   \begin{game}{3}{3}[][]
   	  &  D     &  E   & F  \\
   	 D  &    $0.16$      & $-0.05$ & $-0.08$\\
   	 E &  $-0.64$ & $0.17$ & $0.31$\\
   	 F & $-0.79$ & $-1.06$ & $0.37$\\
   \end{game}
   \caption{Second counterpart game of the Leduc poker empirical game.}
   \label{fig:CP-PSRO2}
\end{minipage}
\end{table}

Again we can now analyse the landscape of equilibria of this game using the introduced theorems. Since the Leduc poker empirical game is asymmetric we need two populations for the asymmetric replicator equations. As mentioned before, analysing and plotting the evolutionary asymmetric replicator dynamics now quickly becomes very tedious as we deal with two simplices, one for each player. More precisely, if we consider a strategy for one player in its corresponding simplex, and that player is adjusting its strategy, will immediately cause the trajectory in the second simplex to change, and vice versa. Consequently, it is not straightforward anymore to analyse the dynamics and equilibrium landscape for both players, as any trajectory in one simplex causes the other simplex to change. A movie illustrates what is meant: we show how the dynamics of player 2 changes in function of player 1. We overlay the simplex of the second player with the simplex of the first player; the yellow dots indicate what the strategy of the first player is. The movie then shows how the dynamics of the second player changes when the yellow dot changes, see \url{https://youtu.be/10m0f3iBECc}.

In order to facilitate the process of analysing this game we can apply the counterpart RD theorems here to remedy the problem, and consequently analyse the game in the far simpler symmetric counterpart games that will shed light onto the equilibrium landscape of the Leduc Poker empirical game.

In Figures \ref{fig:dfieldPSROP1} and \ref{fig:dfieldPSROP2} we show the evolutionary dynamics of the counterpart games. As can be observed in Figure \ref{fig:dfieldPSROP1} the first counterpart game has only one equilibrium, i.e., a mixed Nash equilibrium at the face formed by $A$ and $C$, which absorbs the entire strategy space. Looking at Figure \ref{fig:dfieldPSROP2} we see the situation is a bit more complex in the second counterpart game, here we observe three Nash equilibria: one pure at strategy $D$, one pure at strategy $F$, and one unstable mixed equilibrium at the 1-face formed by strategies $D$ and $F$. Note there is also a rest point at the face formed by strategies $D$ and $E$, which is not Nash. Given that there are no mixed equilibria with full support in both games we cannot apply Theorem \ref{the:one}. Using Theorem \ref{the:two} we now know that we only maintain the two mixed equilibria, i.e. $(0.32,0,0.68)$ (CP1) and $(0.83,0,0.17)$ (CP2), forming the mixed Nash equilibrium $(x=(0.83,0,0.17),y=(0.32,0,0.68))$ of the asymmetric Leduc poker empirical game. The other equilibria in the second counterpart game can be discarded as candidates for Nash equilibria in the Leduc poker empirical game since they also do not appear for player 1 when we permute the strategies for player 1 (not shown here).

\begin{figure}[!tbp]
  \centering
  \begin{minipage}[b]{0.45\textwidth}
     \includegraphics[width=9cm]{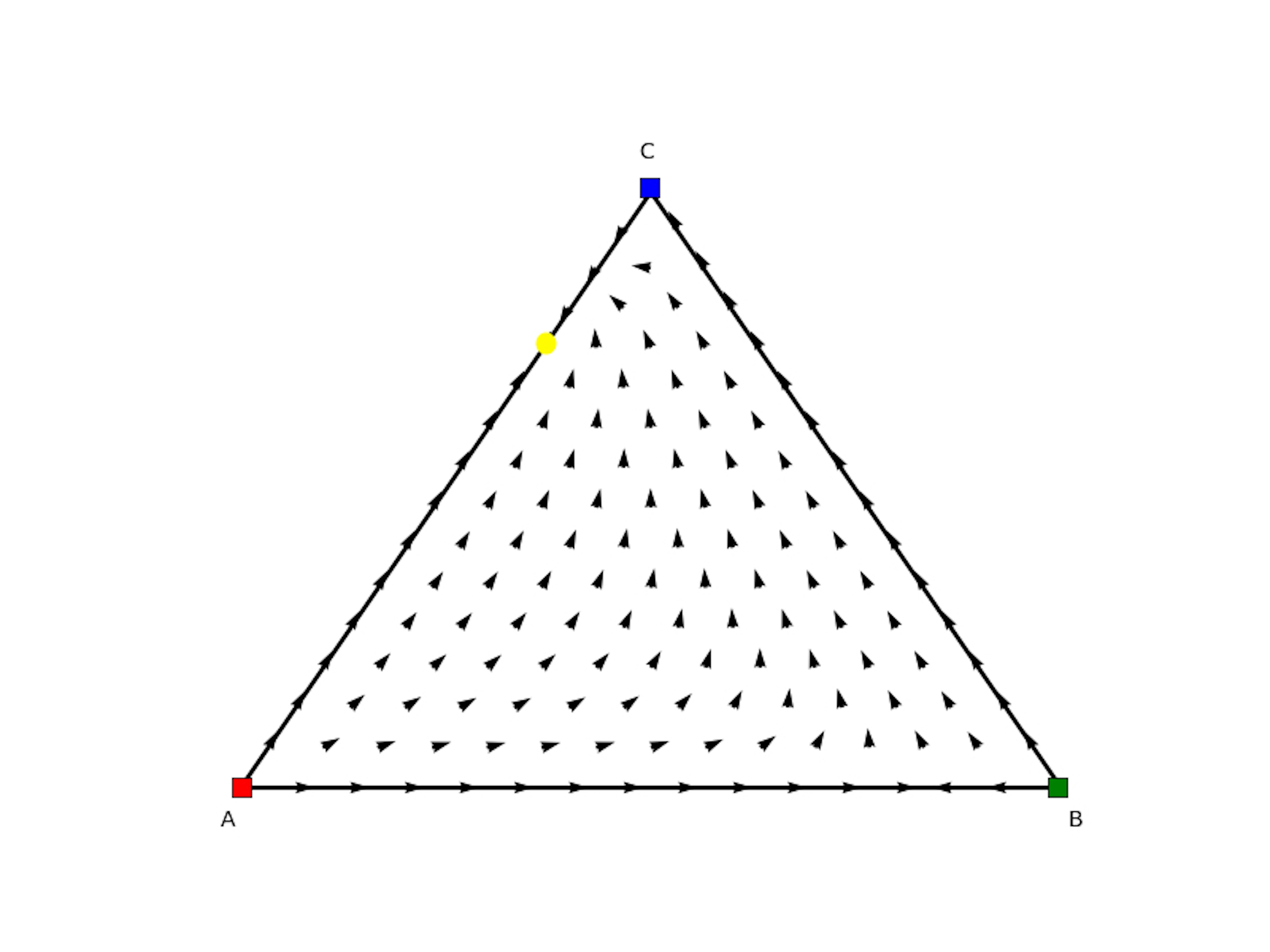}
     \caption{Legend. Directional field plot $\Sigma_3$ of the first counterpart game of the Leduc poker empirical game under study.}
    \label{fig:dfieldPSROP1}
  \end{minipage}
 % \hfill
 \qquad
  \begin{minipage}[b]{0.45\textwidth}
\includegraphics[width=9cm]{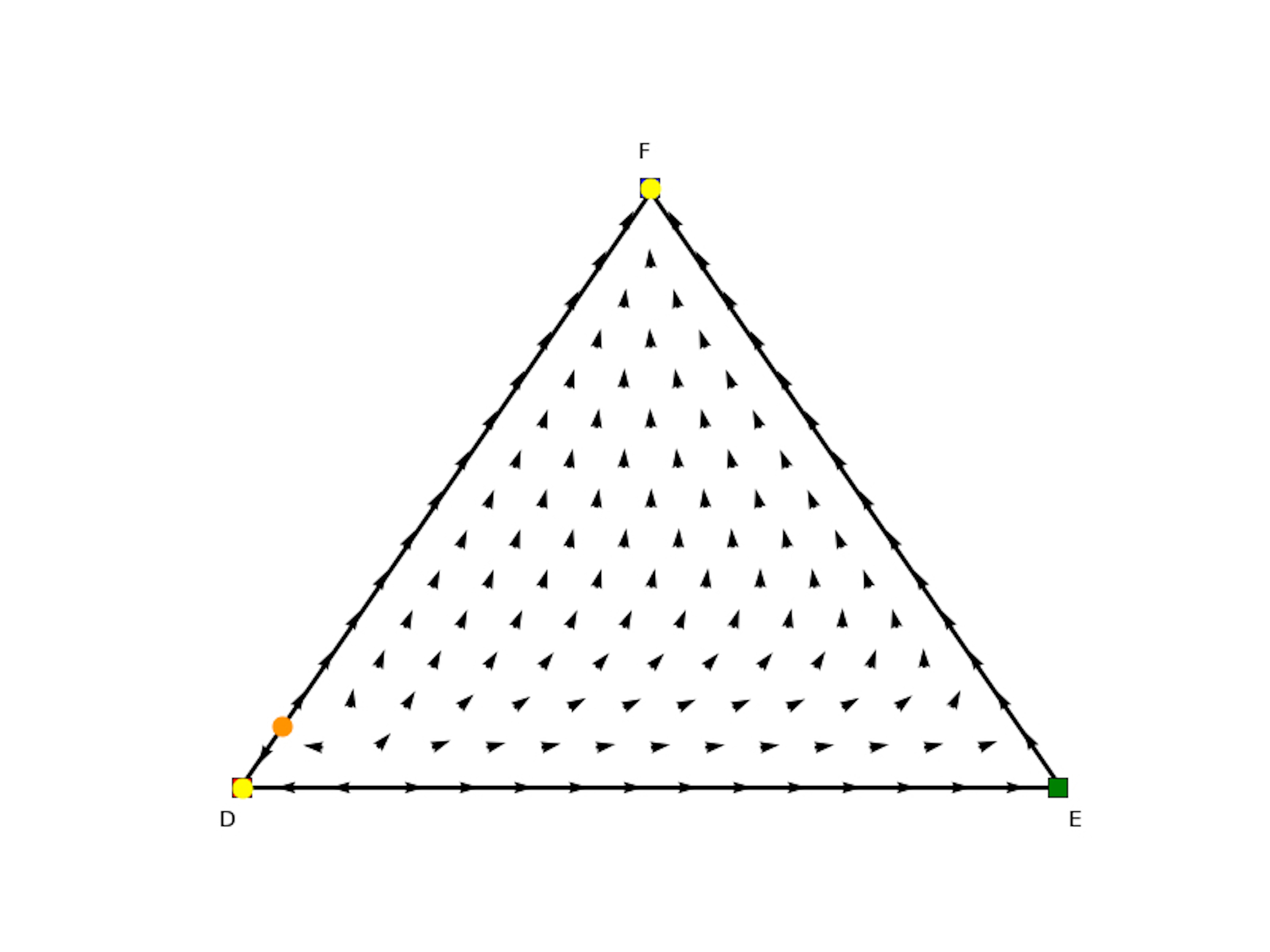}
   \caption{Legend. Directional field plot $\Sigma_3$ of the second counterpart game of the Leduc poker empirical game under study.}
    \label{fig:dfieldPSROP2}
  \end{minipage}
\end{figure}

\subsubsection*{Mixed equilibrium of full support}

As a final example to illustrate the introduced theory, we examine an asymmetric game, that has one completely mixed equilibrium and several equilibria in its counterpart games. The bimatrix game $(A,B)$ is illustrated in Table \ref{fig:asymgamefull} and its symmetric counterparts are shown in Tables \ref{fig:CP-asym1} and \ref{fig:CP-asym2}. 
\begin{table}[h!]
	\centering
   \begin{game}{3}{3}[][]
   	 & D & E & F  \\
    A & $1,0$ & $0,1$ & $2,0$ \\
    B & $0,1$ & $2,0$ & $0,0$  \\
    C & $2,0$ & $0,0$ & $1,1$  \\
   \end{game}
\caption{Payoff matrix of an asymmetric game with mixed equilibrium of full support.}
 \label{fig:asymgamefull}
\end{table}

The asymmetric game has a unique completely mixed Nash equilibrium with different mixtures for the two players, i.e., $(x=(\frac{1}{3},\frac{1}{3},\frac{1}{3});y=(\frac{2}{7},\frac{3}{7},\frac{2}{7}))$.

\begin{table}[h!]
	\centering
\begin{minipage}{.48\textwidth}
   \begin{game}{3}{3}[][]
   	    &  A     &  B   & C  \\
   	 A  &    $1$      & $0$ & $2$\\
   	 B &  $0$ & $2$ & $0$\\
   	 C & $2$ & $0$ & $1$\\
   \end{game}
   \caption{First counterpart game of the asymmetric game.}
   \label{fig:CP-asym1}
\end{minipage}%
\hfill
\begin{minipage}{.48\textwidth}
   \begin{game}{3}{3}[][]
   	  &  D     &  E   & F  \\
   	 D  &    $0$      & $1$ & $0$\\
   	 E &  $1$ & $0$ & $0$\\
   	 F & $0$ & $0$ & $1$\\
   \end{game}
   \caption{Second counterpart game of the asymmetric game.}
   \label{fig:CP-asym2}
\end{minipage}
\end{table}

The two symmetric counterpart game each have seven equilibria. Counterpart game 1 (Table \ref{fig:CP-asym1}), has the following set of Nash equilibria:
$\{\textcolor{blue}{(a)} (p_1=(\frac{2}{7},\frac{3}{7},\frac{2}{7}),p_2=(\frac{2}{7},\frac{3}{7},\frac{2}{7})),(p_1=(\frac{1}{2},\frac{1}{2},0),p_2=(0,\frac{1}{2},\frac{1}{2})),(p_1=(1,0,0),p_2=(0,0,1)),\textcolor{blue}{(b)} (p_1=(0,1,0),p_2=(0,1,0)),(\textcolor{blue}{(c)} p_1=(\frac{1}{2},0,\frac{1}{2}),p_2=(\frac{1}{2},0,\frac{1}{2})),(p_1=(0,0,1),p_2=(1,0,0)),(p_1=(0,\frac{1}{2},\frac{1}{2}),p_2=(\frac{1}{2},\frac{1}{2},0))\}$. Note that there are also two rest points, which are not Nash, at the faces formed by $A$ and $B$ and $B$ and $C$.
From these seven equilibria only (a), (b) and (c) are of interest since these are symmetric equilibria in which both players play with the same strategy (or support). Also counterpart game 2 has seven equilibria, i.e., $\{\textcolor{blue}{(d)} (p_1=(\frac{1}{3},\frac{1}{3},\frac{1}{3}),p_2=(\frac{1}{3},\frac{1}{3},\frac{1}{3})),(p_1=(\frac{1}{2},0,\frac{1}{2}),p_2=(0,\frac{1}{2},\frac{1}{2})),(\textcolor{blue}{(e)} p_1=(0,0,1),p_2=(0,0,1)),(p_1=(0,\frac{1}{2},\frac{1}{2}),p_2=(\frac{1}{2},0,\frac{1}{2})),(\textcolor{blue}{(f)} p_1=(\frac{1}{2},\frac{1}{2},0),p_2=(\frac{1}{2},\frac{1}{2}),0),(p_1=(1,0,0),p_2=(0,1,0)),(p_1=(0,1,0),p_2=(1,0,0)\}$ 
of which only (d), (e) and (f) are of interest.

We observe that only the completely mixed equilibrium of the asymmetric game, i.e., $(x=(\frac{1}{3},\frac{1}{3},\frac{1}{3});y=(\frac{2}{7},\frac{3}{7},\frac{2}{7}))$, has its counterpart in the symmetric games. To apply the theorems we only need to have a look at equilibria (a), (b) and (c) in counterpart game 1, and (d), (e) and (f) in counterpart game 2. These equilibria can also be observed in the directional field plots, respectively, trajectory plots, illustrating the evolutionary dynamics of both counterpart games in Figures \ref{fig:dfieldasymP1},   \ref{fig:tplotasymP1}, \ref{fig:dfieldasymP2} and \ref{fig:tplotasymP2}. %These show the evolutionary dynamics of both counterpart games, and Figures \ref{fig:tplotasymP1} and \ref{fig:tplotasymP2} are the corresponding trajectory plots. 
Figure \ref{fig:dfieldasymP1} visualises the three remaining equilibria (a), (b) and (c), with (a) indicated as a yellow oval, and (b) and (c) both indicated as green ovals. As can be observed, (a) is an unstable mixed equilibrium, (b) is a stable pure equilibrium, and (c) is a partly mixed equilibrium at the 2-face formed by strategies A and C. 

We can make the same observation for the second counterpart game, and see that (d), (e) and (f) are equilibria in Figure \ref{fig:dfieldasymP2}. Equilibrium (d), indicated by a yellow oval, is completely mixed, equilibrium (e) is a pure equilibrium in corner F (green oval), and (f) is a partly mixed equilibrium on the 2-face formed by strategies D and E (green ovals as well).

If we now apply Theorem \ref{the:one} we know that we can combine the mixed equilibria of full support of both counterpart games into the mixed equilibrium of the original asymmetric game, in which the mixed equilibrium of counterpart game 1, i.e. $(\frac{2}{7},\frac{3}{7},\frac{2}{7})$, becomes the part of the mixed equilibrium in the asymmetric game of player 2, and the mixed equilibrium of counterpart game 2, i.e. $(\frac{1}{3},\frac{1}{3},\frac{1}{3})$, becomes the part of the mixed equilibrium in the asymmetric game of player 1, leading to $(x=(\frac{1}{3},\frac{1}{3},\frac{1}{3});y=(\frac{2}{7},\frac{3}{7},\frac{2}{7}))$. Both equilibria are unstable in the counterpart games and also form an unstable mixed equilibrium in the asymmetric game.

%In the example there is 1 equilibrium in the asymmetric version of the game. And there are 7 equilibria per counterpart game.

%For B only 3 of them can be considered since they are single population equilibria.
%For A also only 3 can be considered for the same reason.

%The only one one with a matching support is $()$ and thus we have one Nash equilibrium for the "identity permutation".

\begin{figure}[!tbp]
  \centering
  \begin{minipage}[b]{0.45\textwidth}
     \includegraphics[width=9cm]{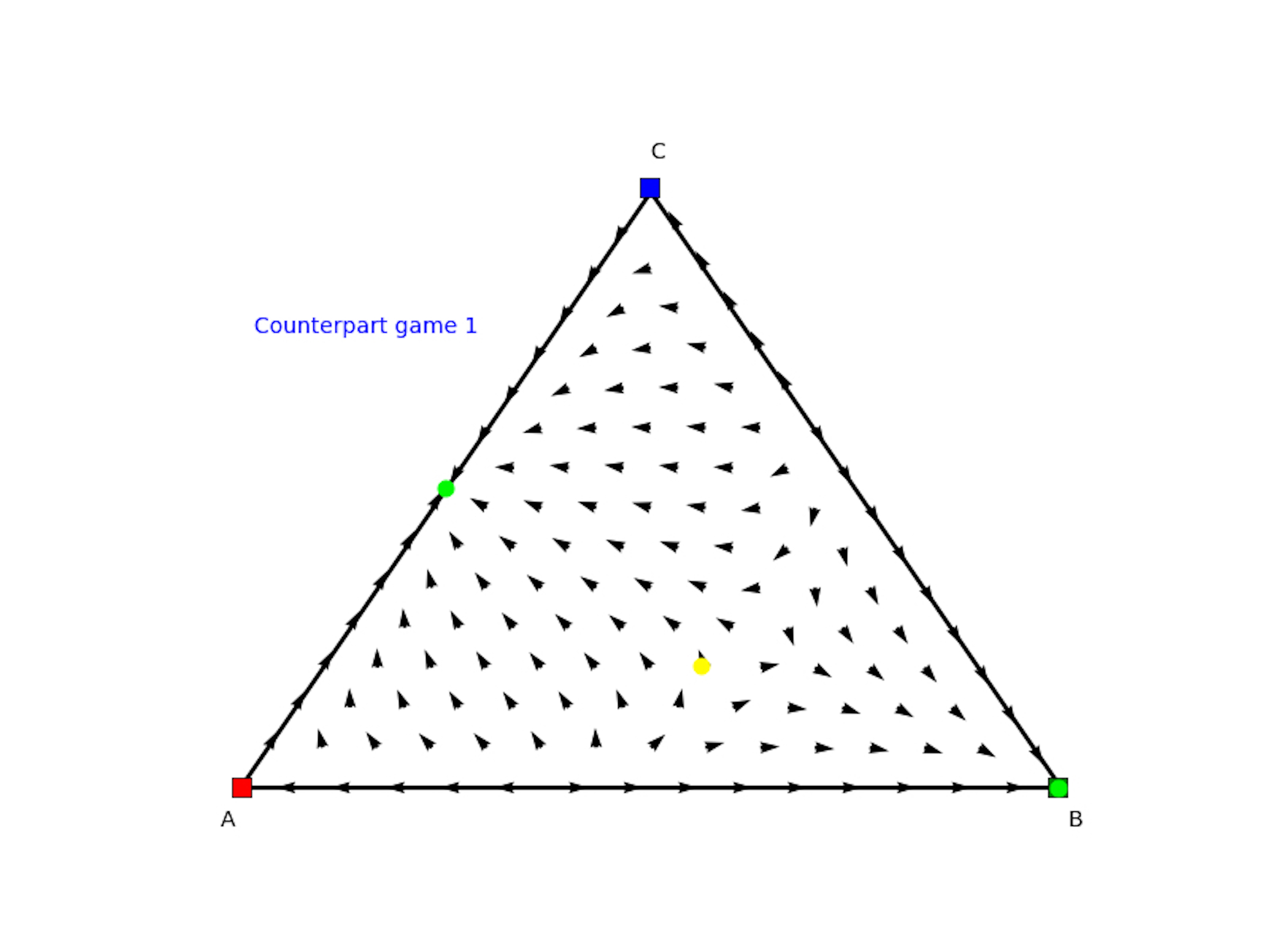}
     \caption{Legend. Directional field plot $\Sigma_3$ of the first counterpart game of the mixed equilibrium asymmetric game.}
    \label{fig:dfieldasymP1}
  \end{minipage}
 % \hfill
 \qquad
  \begin{minipage}[b]{0.45\textwidth}
\includegraphics[width=9cm]{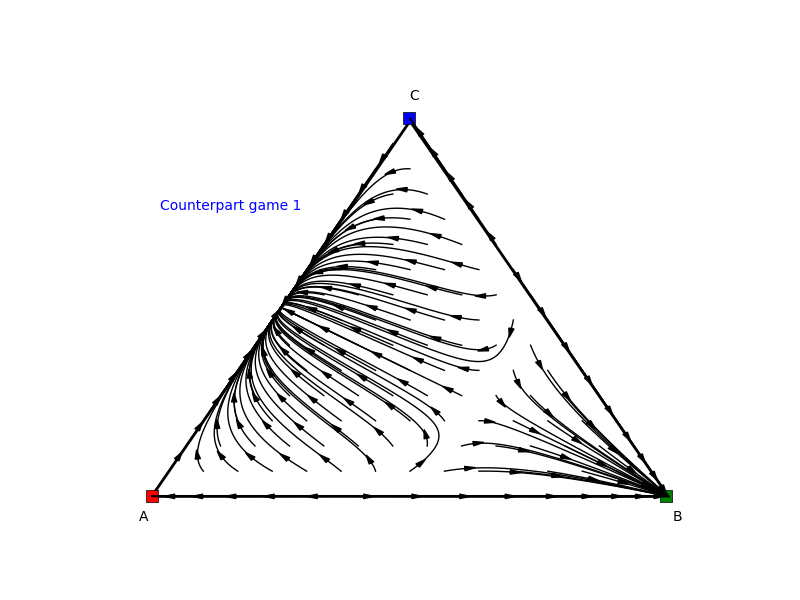}
   \caption{Legend. Trajectory plot $\Sigma_3$ of the first counterpart game of the mixed equilibrium asymmetric game.}
    \label{fig:tplotasymP1}
  \end{minipage}
\end{figure}

\begin{figure}[!tbp]
  \centering
  \begin{minipage}[b]{0.45\textwidth}
     \includegraphics[width=9cm]{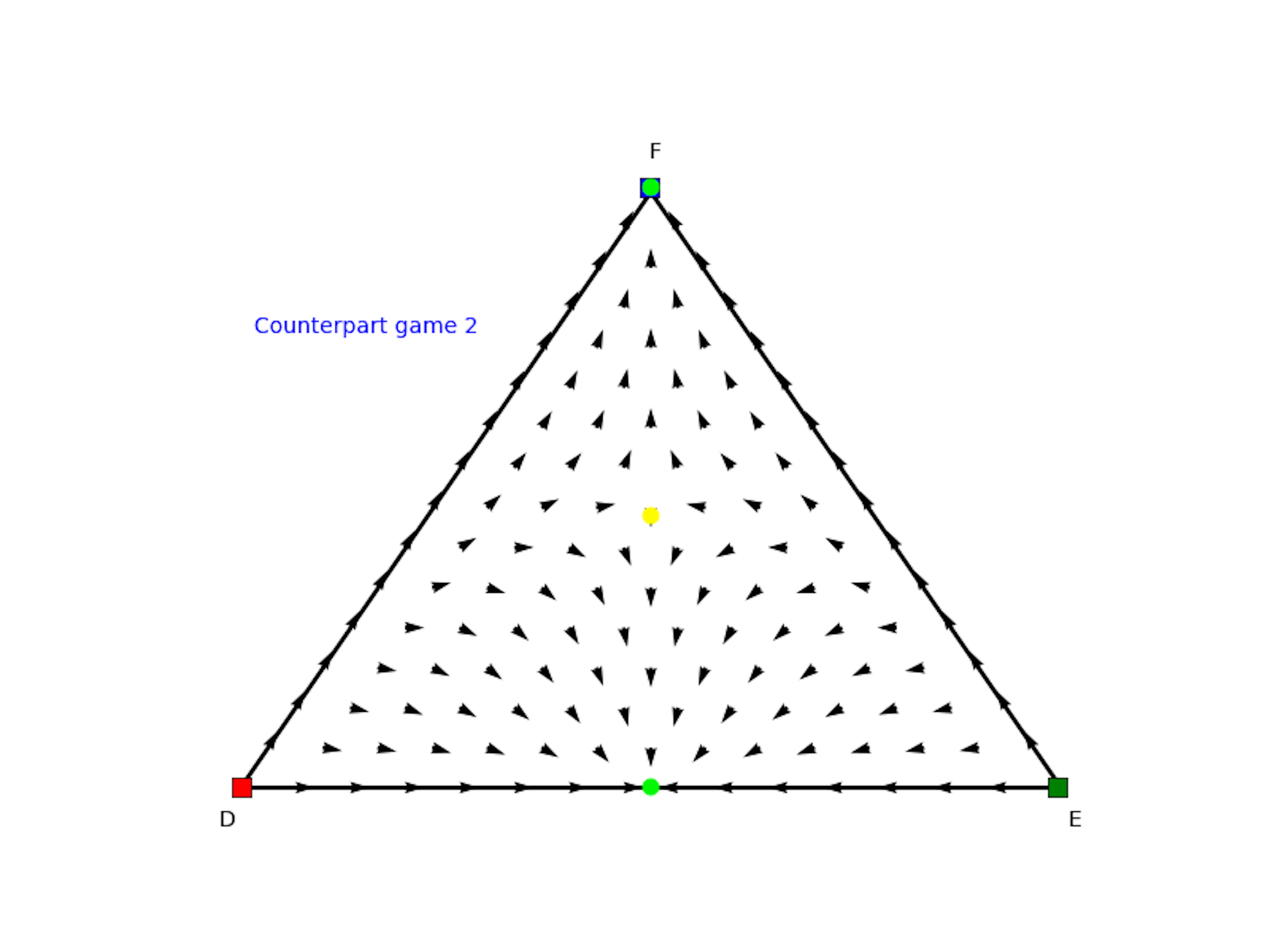}
     \caption{Legend. Directional field plot $\Sigma_3$ of the second counterpart game of the mixed equilibrium asymmetric game.}
    \label{fig:dfieldasymP2}
  \end{minipage}
 % \hfill
 \qquad
  \begin{minipage}[b]{0.45\textwidth}
\includegraphics[width=9cm]{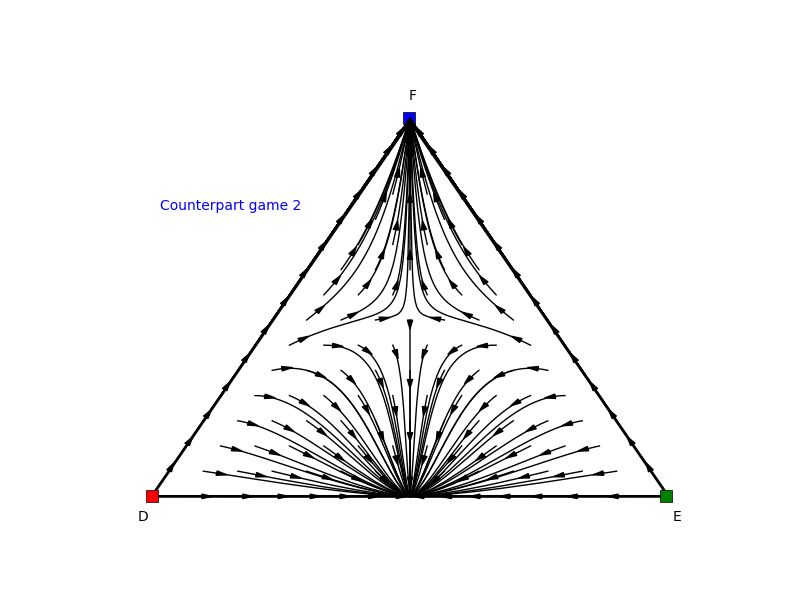}
   \caption{Legend. Trajectory plot $\Sigma_3$ of the second counterpart game of the mixed equilibrium asymmetric game.}
    \label{fig:tplotasymP2}
  \end{minipage}
\end{figure}
%To permute the actions, you need to permute the columns of both A and B. So your permutation results in:
%B_permute=[[1,0,0],[0,1,0],[0,0,1]] and all 7 will be single population equilibria.
%A_permute=[[0,1,2],[2,0,0],[0,2,1]] and there is only one Nash for that one.
%So, we still only have one Nash

\section*{Discussion}

Replicator Dynamics have proved to be an excellent tool to analyse the Nash landscape of multiagent interactions and distributed learning in both abstract games and complex systems~\cite{BloembergenTHK15,Walsh02,Wellman06,TuylsP07}.  
The predominant approach has been the use of symmetric replicator equations, allowing for a relatively straightforward analysis in symmetric games. Many interesting real-world settings though involve roles or player-types for the different agents that take part in an interaction, and as such are \emph{asymmetric} in nature. So far, most research has avoided to carry out RD analysis in this type of interactions, by either constructing a new symmetric game, in which the various actions of the different roles are joined together in one population~\cite{Cressman03,Accinelli2011}, or by considering the various roles and strategies as heuristics, grouped in one population as well~\cite{Walsh02,Walsh03,PhelpsPM04}. In the latter approach the payoffs due to different player-types are averaged over many samples of the player type resulting in a single average payoff to each player for each entry in the payoff table. 

The work presented in this paper takes a different stance by decomposing an asymmetric game into its symmetric counterparts. This method proves to be mathematically simple and elegant, and allows for a straightforward analysis of asymmetric games, without the need for turning the strategy spaces into one simplex or population, but instead allows to keep separate simplices for the involved populations of strategies. Furthermore, the counterpart games allow to get insight in the type and form of interaction of the asymmetric game under study, identifying its equilibrium structure and as such enabling analysis of abstract and empirical games discovered through multiagent learning processes (e.g. Leduc poker empirical game), as was shown in the experimental section. 

A deeper counter-intuitive understanding of the theoretical results of this paper is that when identifying Nash equilibria in the counterpart games with \emph{matching} support (including permutations of strategies for one of the players), it turns out that also the combination of those equilibria form a Nash equilibrium in the corresponding asymmetric game. In general, the vector field for the evolutionary dynamics of one player is a function of the other player's strategy, and hence a vector field in one player's simplex doesn't carry much information as any equilibria you observe in it are changing with time as the other player is moving too. However, if you position the second player at a Nash equilibrium, it turns out that player one becomes indifferent between his different strategies, and remains stationary under the RD. This gives the unique situation in which the vector field plot for the second player's simplex is actually meaningful, because the assumption of player one being stationary actually holds (and vice versa). This is what we end up using when establishing the correspondence of the Nash Equilibria in asymmetric and counterpart games, and why the single-simplex plots for the counterpart games are actually meaningful for the asymmetric game - but this is also why they only describe the Nash Equilibria faithfully, but fail to be a valid decomposition of the full asymmetric game away from equilibrium. 

These findings shed new light on asymmetric interactions between multiple agents and provide new insights that facilitate a thorough and convenient analysis of asymmetric games. As pointed out by Veller and Hayward\cite{Veller}, many real-world situations, in which one aims to study evolutionary or learning dynamics of several interacting agents, are better modelled by asymmetric games. As such these theoretical findings can facilitate deeper analysis of equilibrium structures in evolutionary asymmetric games relevant to various topics including economic theory, evolutionary biology, empirical game theory, the evolution of cooperation, evolutionary language games and artificial intelligence\cite{Baek16,Hilbe17,Allen17,Santos16,Moreira13,Steels00}.

Finally, the results of this paper also nicely underpin what is said in H. Gintis' book on the evolutionary dynamics of asymmetric games, i.e., \emph{'although the static game pits the row player against the column player, the evolutionary dynamic pits row players against themselves and column players against themselves'}\cite{Gintis09} (chapter 12, p.292). He also indicates that this aspect of an evolutionary dynamic is often misunderstood. The use of our counterpart dynamics supports and illustrates this statement very clearly, showing that in the counterpart games species play games within a population and as such show an intra-species survival of the fittest, which is then combined into an equilibrium of the asymmetric game.

\section*{Data availability statement}

No datasets were generated or analysed during the current study.

%\bibliographystyle{naturemag.bst} 
%\bibliography{bib}

%\noindent LaTeX formats citations and references automatically using the bibliography records in your .bib file, which you can edit via the project menu. Use the cite command for an inline citation, e.g. ~\cite{Figueredo:2009dg}.

\section*{Acknowledgements}

We are very grateful to D. Bloembergen and O. Pietquin for helpful comments and discussions.

\section*{Author contributions statement}

K.T. and J.P. designed the research and theoretical contributions. K.T. implemented the experimental illustrations. K.T., J.P. and M.L. performed the simulations. All authors analysed the results and wrote and reviewed the paper.

\section*{Additional information}

The authors declare no competing financial interests.

%\begin{figure}[ht]
%\centering
%\includegraphics[width=\linewidth]{stream}
%\caption{Legend (350 words max). Example legend text.}
%\label{fig:stream}
%\end{figure}

\end{document}